\documentclass[11pt]{article}
\usepackage{a4wide}

\usepackage{amsmath}
\usepackage{amssymb}
\usepackage{graphicx}
\usepackage{subcaption}
\usepackage{url}
\usepackage{comment}
\usepackage{amsfonts}
\usepackage{amsthm}
\usepackage{comment}
\usepackage{thmtools}
\usepackage{thm-restate}
\usepackage{mdframed}

\newtheorem{theorem}{Theorem}[]

\usepackage{todonotes}

\usepackage{graphicx}
\usepackage{booktabs}
\usepackage{pdflscape}
\usepackage{mdframed}
\usepackage{mathtools}
\usepackage{pdflscape}

\DeclareMathOperator{\Seq}{Seq}

\newcommand{\eps}{\varepsilon}
\renewcommand{\P}{\textsc{P}}
\newcommand{\NP}{\textsc{NP}}

\usepackage{xspace}
\usepackage{framed}

\usepackage{xcolor}
\usepackage{colortbl}
\colorlet{tableheadcolor}{gray!25} 
\colorlet{tablerowcolor}{gray!20} 
\newcommand{\rowcol}{\rowcolor{tablerowcolor}} %

\title{Learning fine-grained search space pruning\\and heuristics for combinatorial optimization~\footnote{The present manuscript integrates and extends the works~\cite{our-nips,Grassia2019,Dutta2019}, which appeared at AAAI'19, the DSO Workshop at IJCAI'19 and CIKM'19, respectively. Part of this work was done while all authors were at Nokia Bell Labs, Ireland.}}

\author{Juho Lauri \and Sourav Dutta\thanks{Eaton Corp., Ireland} \and Marco Grassia\thanks{University of Catania, Italy} \and Deepak Ajwani\thanks{University College Dublin, Ireland}}

\begin{document}
\maketitle

\begin{abstract}
  Combinatorial optimization problems arise naturally in a wide range of applications from diverse domains. Many of these problems are NP-hard and designing efficient heuristics for them requires considerable time, effort and experimentation. On the other hand, the number of optimization problems in the industry continues to grow. In recent years, machine learning techniques have been explored to address this gap. In this paper, we propose a novel framework for leveraging machine learning techniques to scale-up \emph{exact} combinatorial optimization algorithms.
  In contrast to the existing approaches based on deep-learning, reinforcement learning and restricted Boltzmann machines that attempt to directly learn the output of the optimization problem from its input (with limited success), our framework learns the relatively simpler task of pruning the elements in order to reduce the size of the problem instances. In addition, our framework uses only interpretable learning models based on intuitive local features and thus the learning process provides deeper insights into the optimization problem and the instance class, that can be used for designing better heuristics.

 For the classical maximum clique enumeration problem, we show that our framework can prune a large fraction of the input graph (around 99~\% of nodes in case of sparse graphs) and still detect almost all of the maximum cliques. This results in several fold speedups of state-of-the-art algorithms. Furthermore, the classification model used in our framework highlights that the chi-squared value of neighborhood degree has a statistically significant correlation with the presence of a node in a maximum clique, particularly in dense graphs which constitute a significant challenge for modern solvers. We leverage this insight to design a novel heuristic for this problem outperforming the state-of-the-art. Our heuristic is also of independent interest for maximum clique detection and enumeration.
\end{abstract}

\section{Introduction}
Combinatorial optimization is at the heart of a large number of applications from a wide range of domains such as economics (e.g., price optimization~\cite{FLS15}, efficient energy scheduling~\cite{NAL16}), bioinformatics (e.g.,~\cite{BW06,DKRJ13,KM95}), robotics (e.g.,~\cite{SR13}), industrial production, and planning (e.g.,~\cite{R12,R14}). In fact, numerous real-life decision-making problems have been formulated in terms of combinatorial optimization problems~\cite{Trevisan11,Korte12} and as a result, combinatorial optimization algorithms are widely applied in industry.

Combinatorial optimization problems typically involve finding groupings (subsets), orderings or assignments of a discrete, finite set of objects that satisfy certain conditions or constraints. For instance, in the maximum clique enumeration problem, the goal is to explicitly list all the largest subsets of nodes that are all adjacent to each other. In the travelling salesman problem (TSP), the goal is to identify a subset of edges that constitute a smallest tour covering all nodes (or alternatively, an ordering of nodes which results in a shortest tour). Both these and many other computational problems are $\NP$-hard, implying that --- unless $\P = \NP$ --- no polynomial-time algorithms exist for these problems that can solve every instance of the problem to optimality. Over the last century, numerous approaches have been developed for these applications, including (i) exact algorithms with exponential time complexity, (ii) approximation algorithms with formal guarantees on the solution quality, (iii) parameterized algorithms (see e.g.,~\cite{fpt-book} for more), (iv) carefully designed heuristics that leverage the structure often available in real-world instances and (v) meta-heuristic frameworks such as genetic algorithms or ant colony optimization. While exact algorithms have poor scalability, the design of approximation algorithms, parameterized algorithms and domain-specific heuristics require considerable development and design time. Moreover, it is not always the case that theoretically appealing approximation or parameterized algorithms would be practical. Similarly, meta-heuristics often require significant configuration time to select the best parameters and operators for a given optimization problem and can take considerable time to find a combination of elements close to the optimal solution. Furthermore, the solutions produced by heuristics (including those generated by meta-heuristic frameworks) can be arbitrarily far from an optimal solution.

As the number of optimization problems continues to grow in industry and the design time for efficient solutions remains high, it is vital to explore if we can accelerate the algorithm design process using machine learning techniques. With the recent advances in machine learning techniques and its success in areas such as multi-media classification, machine translation, text generation, recommender systems and computer games, researchers have started exploring if they can also be successfully applied to combinatorial optimization.

Existing machine learning techniques for combinatorial optimization can be broadly categorized into three different approaches:
\begin{enumerate}
\item
  Supervised deep-learning approaches that directly learn the output from the input. Examples of this framework include the pointer network~\cite{Vinyals2015}. 
\item
  Reinforcement learning to learn a mapping from state to policy. Examples of this framework include neural combinatorial optimization~\cite{Bello2016} and greedy Q-learning for graph optimization problems~\cite{Khalil2017}.
\item
  Unsupervised approaches based on restricted Boltzmann machines. Examples of this framework include the Estimation of Distribution Algorithm by Probst~{et al.}~\cite{PRG17}.
\end{enumerate}

These frameworks aim to learn the exact decision boundary to separate the elements in the optimal subset from the remaining elements. Since many combinatorial optimization problems are NP-hard, this is a challenging task requiring complex learning models with a large number of parameters. As a result, the learning models are not easy to interpret. Since the learned heuristic (mapping from input to output) is implicit in the complex learning model, this implies that the learned heuristic is itself not interpretable. This has the following consequences:
\begin{itemize}
\item
  With a large number of parameters, the learned algorithm is not easy for humans to understand and consequently there is little potential for mathematical analysis of the resultant algorithm.
\item
  It is not clear if the learned model will still work if there is an additional constraint added to the problem. This is a major concern for applications in industry where the first modelling of a problem into the optimization objective and associated constraints is rarely enough and new constraints are incrementally discovered and added.
\item
  It is not clear if the learned model will still work if a dataset from a slightly different distribution is used as an input. This has further implications for cross-domain generalizability of the learning model.
  \end{itemize}

\begin{figure}[t]
  \begin{subfigure}{0.48\textwidth}
  \begin{center}
    \includegraphics[scale=0.6]{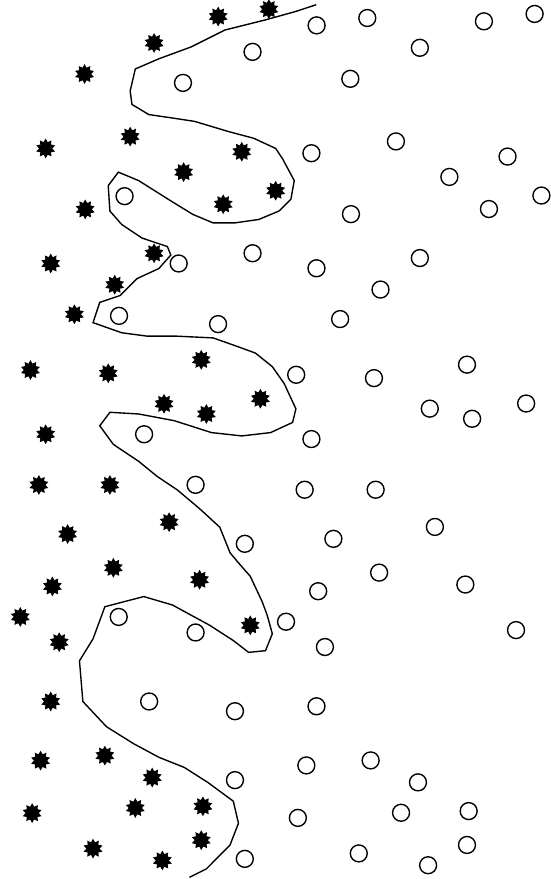}   
  \caption{\label{fig:exact_decision} Exact decision boundary}
    \end{center}
  \end{subfigure}
  \begin{subfigure}{0.48\textwidth}
  \begin{center}
    \includegraphics[scale=0.6]{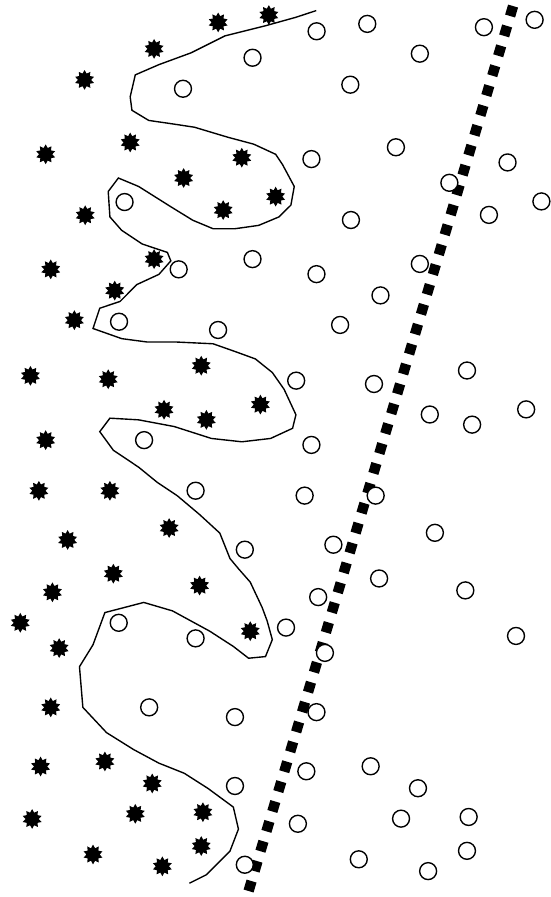}    
  \caption{\label{fig:simple_boundary} Interpretable model boundary}
    \end{center}
  \end{subfigure}
  \begin{subfigure}{0.48\textwidth}
  \begin{center}
    \includegraphics[scale=0.6]{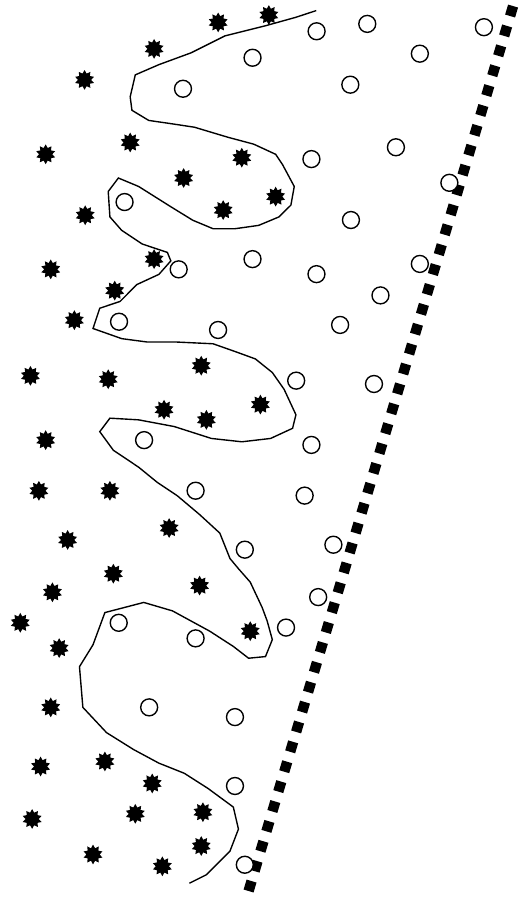}    
  \caption{\label{fig:simple_boundary_prune} Pruning}
    \end{center}
  \end{subfigure}
  \begin{subfigure}{0.48\textwidth}
  \begin{center}
    \includegraphics[scale=0.6]{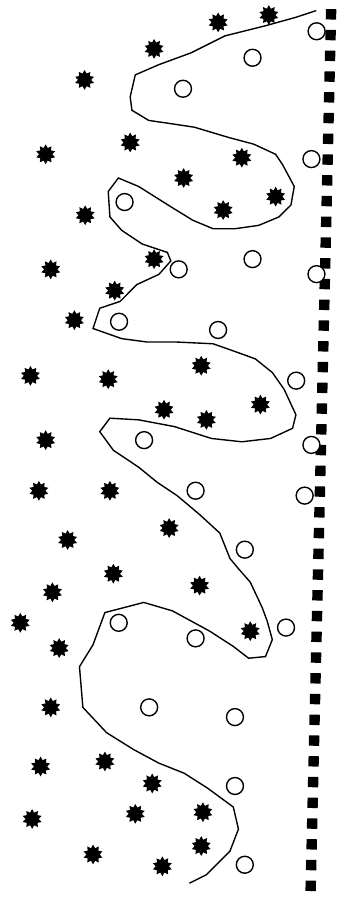}    
  \caption{\label{fig:second_boundary} Repeated pruning}
    \end{center}
  \end{subfigure}
  \caption{\label{fig:framework} Depiction of the overall framework. The black stars indicate the elements in the optimal subset while the white circles represent the elements not in the optimal subset. Earlier learning frameworks attempted to learn the exact decision boundary (drawn with a solid black line). Our framework use simple, interpretable classifier (shown by dashed lines) to repeatedly prune the white circles.}
\end{figure}

In contrast, we propose a novel framework for solving combinatorial optimization problems that uses interpretable learning models based on intuitive local features. For the interpretable models to work well, we focus on the relatively simpler task of (non-exhaustive) pruning of the elements that are \emph{not} in the optimal subset. This reduces the problem size, often significantly, enabling existing solvers to deal with considerably larger instances. Further, it has the potential to not only reduce the size of the instances, but to make the search space easier to handle by breaking symmetries, for example.

To prune the elements further, we extend our framework to have multiple pruning stages. In each stage, the framework learns a new classification model for elements that were not pruned by earlier classification models (thereby increasingly focusing on harder elements to prune). Figure~\ref{fig:framework} illustrates the exact decision boundary and the multi-stage pruning framework.

Furthermore, the learning process in our framework provides deeper insights into the optimization problem and the instance class. In particular, it identifies the combination of simple features that are most indicative of an element belonging to an optimal subset. This insight can be leveraged to design better heuristics for the optimization problems.

 For the classical maximum clique enumeration problem, we show that our framework can prune a large fraction of the graph (around 99~\% of nodes in case of sparse graphs) and still detect almost all of the maximum cliques. This results in several fold speedups of state-of-the-art algorithms. Furthermore, the classification model used in our framework highlights that the chi-squared value of neighborhood degree has a statistically significant correlation with the presence of a node in a maximum clique, particular in the case of dense graphs which constitute a major challenge for state-of-the-art solvers. We leverage this insight to design a novel heuristic for this problem enhancing the state-of-the-art. Our heuristic is also of independent interest for maximum clique detection and enumeration.

\section{Related work}

\subsection{Machine learning for combinatorial optimization} Although the area of learning techniques for combinatorial optimization is only beginning to flourish, many frameworks have been developed in the last five years. DiCarro~\cite{dC19} surveyed various learning frameworks for combinatorial optimization, covering various issues related to the complex architectures of the models and the large number of parameters. The existing literature on learning techniques can be broadly categorized into the following classes.

\paragraph{Supervised learning to directly learn the solution of combinatorial optimization problems} This framework involves the use of deep-learning for learning the output.

Vinyals et al.~\cite{Vinyals2015} viewed the task of learning combinatorial optimization solutions as a sequence-to-sequence learning problem. They aimed for directly learning the output solution from the input sequence for optimization problems such as convex hull and Delauney triangulation. The authors used a recurrent neural network (RNN) for the sequence-to-sequence learning. To deal with the issue of long-range correlations (elements far from each other in the input sequence affecting the same output element), they used an attention mechanism to augment the RNN model. To deal with the issue of a fixed vocabulary size required for the output of a recurrent neural network, they used pointers to elements in the input stream, resulting in the name pointer network.

\paragraph{Reinforcement learning for combinatorial optimization} Deep learning approaches require a large amount of training data and to generate this, a large number of NP-hard problem instances need to be solved, limiting the applicability of these techniques. On the other hand, given a solution, it is relatively easy to evaluate the quality of the solution by computing the optimization objective. Thus, in recent years, reinforcement learning based techniques have been developed to solve optimization problems. In this framework, the goal is to learn a stochastic policy that samples solutions of high quality with high probability. In particular, Khalil et al.~\cite{Khalil2017} used the Q-learning technique to learn the solutions for graph optimization problems, specifically minimum vertex cover and maximum cut. They encode nodes using a graph embedding technique and then build a solution using a greedy construction meta-algorithm. The greedy decisions are based on an estimated Q-function parameterized by the embedding. The embedding parameters for the Q-function are updated step by step based on the partial solution computed. 

The GCOMB approach of Mittal et al.~\cite{MDMRS19} follows the same framework, but claims to scale to very large graphs. Another example of this framework is the use of neural combinatorial optimization~\cite{Bello2016} for TSP.

\paragraph{Unsupervised learning for combinatorial optimization} Unsupervised approaches via restricted Boltzmann machines have also been used to deal with combinatorial optimization problems. An example of this framework is the Estimation of Distribution Algorithm by Probst et al.~\cite{PRG17}. This approach iteratively builds and samples from a probabilistic model of candidate solutions. Intuitively, these approaches build information about the probability distribution of good candidate solutions. This model is built using contrastive divergence.

\paragraph{Limitations of the above frameworks} The learning models used in these existing state-of-the-art frameworks are both hard to interpret and architecturally complex. For instance, the neural combinatorial optimization approach~\cite{Bello2016} is a combination of pointer networks (with two LSTM networks), a Monte Carlo policy gradient and an actor-critic architecture. The complexity of these approaches comes at a significant cost of interpretability. Since the learned algorithm (mapping from input to output) is implicit in the complex learning model, this implies that the learned heuristic is itself not easy for humans to understand.
As a direct consequence, it is difficult to analyze the learned algorithm mathematically.
Moreover, it is unclear what features of the input instances are being exploited by the learned heuristic and on which class of datasets will it perform well.

\paragraph{Maximum clique enumeration}
We instantiate our framework for the maximum clique enumeration (MCE) problem. In this problem, the goal is to list all \emph{maximum} (as opposed to maximal) cliques in a given graph. The maximum clique problem is one of the most heavily-studied combinatorial problems arising in various domains such as in the analysis of social networks~\cite{soc,Fortunato2010,Palla2005,Papadopoulos2012}, behavioral networks~\cite{beha}, and financial networks~\cite{finan}. It is also relevant in clustering~\cite{dynamic,Yang2016} and cloud computing~\cite{Wang2014,Yao2013}. The listing variant of the problem, MCE, is encountered in computational biology~\cite{bio,Eblen2012,Yeger2004,mce} in problems like the detection of protein-protein interaction complex, clustering protein sequences, and searching for common cis-regulatory elements~\cite{protein}.

The computational aspects of the problem are well-studied. Indeed, it is $\NP$-hard to even approximate the maximum clique problem within $n^{1-\epsilon}$ for any $\epsilon > 0$~\cite{Zuckerman2006}.
Furthermore, unless an unlikely collapse occurs in complexity theory, the problem of identifying whether a graph of $n$ vertices has a clique of size $k$ is not solvable in time $f(k) n^{o(k)}$ for any function $f$~\cite{Chen2006}. As such, even small instances of this problem can be non-trivial to solve. Moreover, under reasonable complexity-theoretic assumptions, there is no polynomial-time algorithm that preprocesses an instance of $k$-clique to have only $f(k)$ vertices, where $f$ is any computable function depending solely on $k$ (see e.g.,~\cite{fpt-book}). These results indicate that it is unlikely that an efficient preprocessing method for MCE exists that can reduce the size of input instance drastically while guaranteeing to preserve all the maximum cliques. In particular, it is unlikely that polynomial-time sparsification methods (see e.g.,~\cite{Batson2013}) would be applicable to MCE. This has led researchers to focus on heuristic pruning approaches.

\paragraph{On preprocessing for maximum clique}
For the discussion to follow, it will be useful to recall the concept of a \emph{$k$-core} of a graph $G$.
Here, the $k$-core of $G$ is a maximal subgraph of $G$ where every vertex in the subgraph has degree at least $k$ in the subgraph. The \emph{core number} of a vertex $v$ is the largest $k$ for which a $k$-core containing $v$ exists.
A typical preprocessing step in a state-of-the-art solver is the following: (i) quickly find a large clique (say of size $k$), (ii) compute the core number of each vertex of the input graph $G$, and (iii) delete every vertex of $G$ with core number less than $k-1$. This can be equivalently achieved by repeatedly removing all vertices with degree less than $k$. For example, the solver \texttt{pmc}~\cite{Rossi2015b} -- which is regarded as ``\emph{the} leading reference solver''~\cite{San2016} -- use this as the only preprocessing method. 
However, there are two major downsides to this preprocessing step.
First, it is crucially dependant on $k$, the size of a large clique found. 
Since the maximum clique size is $\NP$-hard to approximate within a factor of $n^{1-\epsilon}$, maximum clique estimates with no formal guarantees are used.
Second and more important, it is typical that even if the estimate $k$ was equal to the size of a maximum clique in $G$, the core number of most vertices could be considerably higher than $k-1$. This is particularly true in the case of dense graphs and it results in little or \emph{no} pruning of the search space. Similarly, other preprocessing strategies (see e.g.,~\cite{Eblen2010} for more discussion) depend on $\NP$-hard estimates of specific graph properties and are not useful for pruning dense graphs.

These facts further motivate the quest for preprocessing methods that (i) are effective on dense graphs and (ii) work independently of any estimate~$k$ for the maximum clique size. Unfortunately, as described earlier, under widely-believed complexity-theoretic assumptions, no methods exist that can give strong guarantees for pruning arbitrary graphs. This raises the question if one can discover heuristic methods that can do significant pruning, in practice, on graphs from different domains, including dense graphs. Even more importantly, can we \emph{learn} a heuristic to prune the input instance?

\section{Our framework}
In this section, we describe our framework for subset-based combinatorial optimization problems. For ease of exposition, we describe the framework in terms of the MCE problem.
We stress that our approach is not restricted to MCE, but can be applied to other problems as well.

In our case, we assume the instance is represented as an undirected graph $G=(V,E)$.
Moreover, in contrast to previous approaches, we view \emph{individual vertices} of $G$ as classification problems as opposed to $G$ itself.
That is, the problem is to induce a mapping $\gamma : V \to \{0,1\} $ from a set of $L$ training examples $T = \{ \langle f(v_i), y_i \rangle \}^L_{i=1}$, where $v_i \in V$ is a vertex, $y_i \in \{0,1\}$ a class label, and $f : V \to \mathbb{R}^d$ a mapping from a vertex to a $d$-dimensional feature space.
For reasons of scalability, we strive to keep $d$ small and to ensure that $f$ can be computed efficiently.

\paragraph{Single-stage sparsification} To learn the mapping $\gamma$ from $T$, we use a probabilistic classifier $P$ which outputs a probability distribution over $\{0,1\}$ for a given $f(u)$ for $u \in V$. 
A natural parameterized search strategy, which we call \emph{probabilistic preprocessing} (or \emph{single-stage sparsification}), for enhancing a search algorithm $\mathcal{A}$ by $P$ is as follows. Define a \emph{confidence threshold} $q \in [0,1]$. Delete from $G$ each vertex predicted by $P$ to \emph{not} be in a solution with probability at least $q$, i.e., let $G' = G \setminus V'$, where $V' = \{ u \mid u \in V \wedge P(u = 0) \geq q \}$. 
Execute $\mathcal{A}$ with $G'$ as input instead of $G$.
Here, the purpose of $q$ is to control the error and pruning rate of preprocessing: (i) it is more acceptable to not remove a vertex that is \emph{not} in a solution than to remove a vertex that \emph{is} in a solution, and (ii) a lower value of $q$ translates to a possibly higher pruning rate.
Clearly, this strategy is a heuristic, i.e., it is possible that the cost of an optimal solution in $G'$ differs from that in $G$.

\paragraph{Multi-stage sparsification}
A natural generalization of the probabilistic preprocessing strategy is the following approach that we call \emph{multi-stage sparsification}: 
Let $\mathcal{G}_1$ be the input set of networks.
Consider a graph $G \in \mathcal{G}_1$. 
Let $\mathcal{M}$ be the set of all maximum cliques of $G$, and denote by $V(\mathcal{M})$ the set of all vertices in $\mathcal{M}$.
The positive examples in the training set $T_1$ consist of all vertices that are in some maximum clique ($V(\mathcal{M})$) and the negative examples are the ones in the set $V \setminus V(\mathcal{M})$. 
Since the training dataset can be highly skewed, we under-sample the larger class to achieve a balanced training data. A probabilistic classifier $P_1$ is trained on the balanced training data in stage $1$. Then, in the next stage, we remove all vertices that were predicted by $P_1$ to be in the negative class with a probability above a predefined threshold $q$. We focus on the set $\mathcal{G}_2$ of subgraphs (of graphs in $\mathcal{G}_1$) induced on the remaining vertices and repeat the above process. The positive examples in the training set $T_2$ consists of all vertices in some maximum clique ($V(\mathcal{M})$) and the negative examples are the ones in the set $V \setminus V(\mathcal{M})$, training dataset is balanced by under-sampling and we use that balanced dataset to learn the probabilistic classifier $P_2$. We repeat the process for $\ell$ stages.

\section{Computational features}
\label{sect:features}
In this section, we describe the computational features used in our framework.

\paragraph{Graph-theoretic features}
We use the following graph-theoretic features: \textbf{(F1)} number of vertices, \textbf{(F2)} number of edges, \textbf{(F3)} vertex degree, \textbf{(F4)} local clustering coefficient (LCC), and \textbf{(F5)} eigencentrality.

The crude information captured by features (F1)-(F3) provide a reference for the classifier for generalizing to different distributions from which the graph might have been generated. 
Feature (F4), the LCC of a vertex is the fraction of its neighbors with which the vertex forms a triangle, encapsulating the well-known small world phenomenon.  
Feature (F5) eigencentrality represents a high degree of connectivity of a vertex to other vertices, which in turn have high degrees as well. 
The \emph{eigenvector centrality} $\vec{v}$ is the eigenvector of the adjacency matrix $A$ of $G$ with the largest eigenvalue $\lambda$, i.e., it is the  solution of $\vec{A}\vec{v} = \lambda\vec{v}$.
The $i$th entry of $\vec{v}$ is the \emph{eigencentrality} of vertex $v$. In other words, this feature provides a measure of local ``denseness''. A vertex in a dense region shows higher probability of being part of a large clique.

\paragraph{Statistical features}
In addition, we use the following statistical features: \textbf{(F6)} the $\chi^2$ value over vertex degree, \textbf{(F7)} average $\chi^2$ value over neighbor degrees, \textbf{(F8)} $\chi^2$ value over LCC, and \textbf{(F9)} average $\chi^2$ value over neighbor LCCs.

The intuition behind (F6)-(F9) is that for a vertex $v$ present in a large clique, its degree and LCC would deviate 
from the underlying expected distribution characterizing the graph. 
Further, the neighbors of $v$ also present in the clique would demonstrate such behaviour.
Indeed, statistical features have been shown to be robust in approximately capturing local structural patterns~\cite{graph}.

Statistical significance is captured by the notion of p-value~\cite{fitStatistics}, and well-estimated~\cite{pear} by the {\em Pearson's chi-square statistic}, $\chi^2$, computed as 
\begin{equation}
\label{eq:chis}
\chi^2 = \sum_{\forall i}\left[\left(O_i - E_i\right)^2 / E_i\right],
\end{equation}
where $O_i$ and $E_i$ are the observed and expected number of occurrences of the possible outcomes~$i$.

\begin{figure}[t]
    \centering
        \includegraphics[width=0.5\textwidth,keepaspectratio]{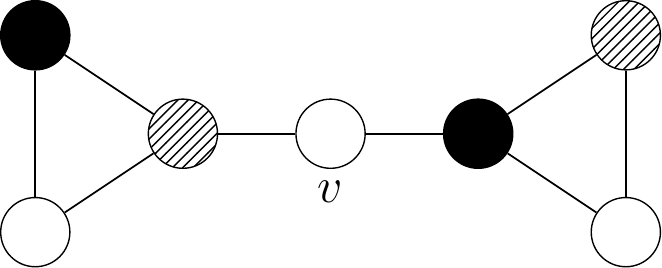} 
    \caption{While the shown proper 3-coloring is optimal, we can swap the non-white colors in either triangle to see that the local chromatic density $\chi_d(v) = 1/3$.}
    \label{fig:chrom-density}
\end{figure}

\paragraph{Local chromatic density}
Let $G=(V,E)$ be a graph.
A \emph{$k$-coloring} of $G$ is a function $c:V\rightarrow \{1, 2, \ldots, k\}$.
A {\it coloring} is a $k$-coloring for some $k \le n$, where $n = |V|$. 
A coloring $c$ is \emph{proper} if $c(u) \neq c(v)$ for every edge $uv\in E$. 
The {\it chromatic number} of $G$, denoted by  $\chi(G)$, is the smallest $k$ such that $G$ has a proper $k$-coloring.
We define the \emph{local chromatic density} of a vertex $v \in V$, denoted by $\chi_d(v)$, as the ratio between the minimum number of distinct colors appearing in $N(v)$ among any optimal proper colouring of $G$ and the chromatic number of $G$.
Informally, the local chromatic density of $v$ is the minimum possible number of colors in the immediate neighborhood of $v$ in any optimal proper coloring of $G$ (see Figure~\ref{fig:chrom-density}).

We use the local chromatic density as the feature \textbf{(F10)}.
A vertex $v$ with high $\chi_d(v)$ means that the neighborhood of $v$ is dense, as it captures the adjacency relations between the vertices in $N(v)$.
Thus, a vertex in such a dense region has a higher chance of belonging to a large clique.

However, the problem of computing $\chi_d(v)$ is computationally difficult.
In the decision variant of the problem, we are given a graph $G=(V,E)$, a vertex $v \in V$, and a ratio $q \in (0,1)$. 
The task is to decide whether there is proper $k$-coloring $c$ of $V$ witnessing $\chi_d(v) \geq q$.
As shown in the following, the claimed hardness is straightforward to establish.
\begin{theorem}
Given a graph $G=(V,E)$, $v \in V$, and $q \in (0,1)$, it is $\NP$-hard to decide whether $\chi_d(v) \leq q$.
\end{theorem}
\begin{proof}
Let $G$ be an instance of graph $k$-coloring, for any fixed $k \geq 3$.
This problem is well-known to be $\NP$-complete for every $k \geq 3$.
We construct $G' = G \cup K_k$, i.e., $G'$ is the disjoint union of $G$ and a complete graph on $k$ vertices.
Fix $v$ to be an arbitrary vertex of the $K_k$.
We claim that $G$ has a proper $k$-coloring if and only if $\chi_d(v) \leq q$, where $q = \tfrac{k-1}{k}$.

If $G$ admits a proper $k$-coloring, we map $\{1,2,\ldots,k\}$ bijectively to $V(K_k)$, implying that $\chi_d(v) = \tfrac{k-1}{k}$.
On the other hand, a proper $k$-coloring of $G'$ witnessing that $\chi_d(v) = \tfrac{k-1}{k}$ is clearly a proper $k$-coloring when restricted to $G$ as well.
\end{proof}
\noindent Despite its computational hardness, we can estimate $\chi_d(v)$ with the following simple heuristic. Compute a proper coloring for $G$ using e.g., the well-known linear-time greedy heuristic of~\cite{Welsh1967} and then estimate $\chi_d(v)$ as the ratio between the number of colors in $N(v)$ divided by the number of colors used by the greedy coloring algorithm.
Note that we could use other graph coloring heuristics as well (see e.g.,~\cite{Lewis2015} for an overview of the state-of-the-art).

\begin{table}[t]
  \centering
  \small
  \caption{The effect of introducing the feature (F10) the local chromatic density into the feature set. The column ``W/o'' is the vertex classification accuracy of the classifier of~Subsection~\ref{subs:sparse} without (F10) while column ``With'' is the same with (F10).}
  \label{tbl:chrom-feat}
  \def\arraystretch{1}
  \setlength{\tabcolsep}{3pt}
  \begin{tabular}{*{111}{l}}
    \toprule
    \multicolumn{2}{c}{\textbf{bio}} & \multicolumn{2}{c}{\textbf{soc}} & \multicolumn{2}{c}{\textbf{socfb}} & \multicolumn{2}{c}{\textbf{web}} & \multicolumn{2}{c}{\textbf{all}} \\
    \cmidrule(lr){1-2}
    \cmidrule(lr){3-4}
    \cmidrule(lr){5-6}
    \cmidrule(lr){7-8}
    \cmidrule(lr){9-10}
    \textbf{W/o} & \textbf{With} & \textbf{W/o} & \textbf{With} & \textbf{W/o} & \textbf{With} & \textbf{W/o} & \textbf{With} & \textbf{W/o} & \textbf{With} \\
    0.95   & 0.98    & 0.89   & 0.99   & 0.90   & 0.95    & 0.96   & 0.99    & 0.87  & 0.96    \\
    \bottomrule
  \end{tabular}
\end{table}

\paragraph{Learning over edges}
Instead of individual vertices, we can view the framework also over \emph{individual edges}.
In this case, the goal is to find a mapping $\gamma' : E \to \{0,1\}$, and the training set $L'$ contains feature vectors corresponding to edges instead of vertices.
We also briefly explore this direction in this work.

\paragraph{Edge features}
We use the following features (E1)-(E9) for an edge $\{u,v\}$.
\textbf{(E1)} Jaccard similarity is the number of common neighbors of $u$ and $v$  divided by the number of vertices that are neighbors of at least one of $u$ and $v$. 
\textbf{(E2)} Dice similarity is twice the number of common neighbors of $u$ and $v$, divided by the sum of their degrees.
\textbf{(E3)} Inverse log-weighted similarity is the number of common neighbors of $u$ and $v$ weighted by the inverse logarithm of their degrees. Formally, we compute $\sum_{x \in N(u) \cap N(v)} 1 / \log(\deg(x))$.
\textbf{(E4)} Cosine similarity is the number of common neighbors of $u$ and $v$ divided by the geometric mean of their degrees.
The next three features are inspired by the vertex features: \textbf{(E5)} average LCC over $u$ and $v$, \textbf{(E6)} average degree over $u$ and $v$, and \textbf{(E7)} average eigencentrality over $u$ and $v$.
\textbf{(E8)} is the number of length-two paths between $u$ and $v$.
Finally, we use \textbf{(E9)} \emph{local edge-chromatic density}, i.e., the number of distinct colors on the common neighbors of $u$ and $v$ divided by the total number of colors used in any optimal proper coloring.

The intuition behind (E1)-(E4) is well-established for community detection; see e.g.,~\cite{Adamic2003} for more.
For (E8), observe that the number of length-two paths is high when the edge is part of a large clique, and at most $n-2$ when $\{u,v\}$ is an edge of a complete graph on $n$ vertices.
Notice that (E9) could be converted into a deterministic rule: the edge $\{u,v\}$ can be safely deleted if the common neighbors of $u$ and $v$ see less than $k-2$ colors in any proper coloring of the input graph $G$, where $k$ is an estimate for $\omega(G)$.
To our best knowledge, such a rule has not been considered previously in the literature.
Further, notice that there are situations in which this rule \emph{can} be applied whereas the similar vertex rule uncovered from (F10) cannot.
To see this, let $G$ be a graph consisting of two triangles $\{a,b,c\}$ and $\{x,y,z\}$, connected by an edge $\{a,x\}$, and let $k = 3$.
The vertex rule cannot delete $a$ nor $x$, but the described edge rule removes $\{a,x\}$.

\section{Experimental results}
In this section, we describe how multi-stage sparsification is applied to the MCE problem and our computational results.


All experiments are executed on a machine equipped with an Intel Core i7-4770K CPU (3.5 GHz), 8 GB of RAM, running Ubuntu 16.04.

\paragraph{Training and test data}
All our datasets are obtained from Network Repository~\cite{Rossi2015} (available at \url{http://networkrepository.com/}).
We discard all vertex and edge weights and parallel edges (if any) and treat every directed edge as undirected.

For dense networks, we choose a total of 30 networks from various categories with the criteria that the edge density is at least 0.5 in each.
We name this category ``dense''.
The test instances are in Table~\ref{tbl:dense}, chosen based on empirical hardness (i.e., they are solvable in reasonable amount of time).

For sparse networks, we choose our training data from four different categories: 31 biological networks (``bio''), 32 social networks (``soc''), 107 Facebook networks (``socfb''), and 13 web networks (``web'').
In addition, we build a fifth category ``all'' that comprises all networks from the mentioned four categories.
The test instances are in Table~\ref{tbl:pruning}.

\paragraph{Feature computation}
We implement the feature computation in C++, relying on the \texttt{igraph}~\cite{igraph} C graph library.
In particular, our feature computation is single-threaded with further optimization possible.

\paragraph{Domain oblivious training via local chromatic density}
To achieve a high classification accuracy, it is natural to assume that the classifier should be trained with networks coming from the same domain, and that testing should be performed on networks from that domain.
Certainly, some similarity is needed between the two for training to be effective.
For example, sparse networks (say trees) should not be representative of dense networks.
However, we demonstrate in Table~\ref{tbl:chrom-feat} that a classifier can be trained with networks from various domains, yet predictions remain accurate across different domains (see column ``all'').
The accuracy is boosted considerably by the introduction of the local chromatic density (F10) into the feature set (see Table~\ref{tbl:chrom-feat}).
In particular, when generalizing across various domains, the impact on accuracy is almost 10~\%. 
For this reason, rather than focusing on network categories, we only consider networks by edge density (at least 0.5 or not).

\paragraph{Accuracy measures and setup}
For our experiments, the \emph{vertex pruning ratio} is the ratio of the number of vertices removed from the instance to the number of vertices $|V|$ in the original instance.
The \emph{edge pruning ratio} is defined similarly, but for edges instead of vertices.
We say \emph{clique accuracy} is one precisely when the number of \emph{all} maximum cliques of the instance $G$ is equal to the number of all maximum cliques of the reduced instance $G'$ and $\omega(G) = \omega(G')$.

\paragraph{State-of-the-art solvers for MCE}
To our best knowledge, the only publicly available maximum clique solvers able to list all maximum cliques\footnote{For instance, \texttt{pmc}~\cite{Rossi2015b} does not have this feature.} are
\texttt{cliquer}~\cite{Ostergard2002}, based on a branch-and-bound strategy; and \texttt{MoMC}~\cite{Li2017}, introducing incremental maximum satisfiability reasoning to a branch-and-bound strategy. 
We use these solvers in our experiments\footnote{It is worth noticing that in principle, one could solve the problem by any algorithm that lists all \emph{maximal cliques}. However, even  such algorithms solve a more general problem (i.e., every maximum clique is maximal but the opposite is not true in general) which usually comes with a significantly higher computational cost.}.

\subsection{Dense networks}
\label{subs:dense}
In this subsection, we show results for probabilistic preprocessing on dense networks (i.e., edge density at least 0.5).

\begin{landscape}
\begin{table*}[t]
\centering
\small
\caption{Experiments for dense graphs. The column ``$\omega$'' is the max.\ clique size and the column ``n.\ $\omega$'' is the number of such cliques. In both, * means the quantity is preserved in the preprocessed instance; otherwise the new quantity is in parenthesis.
The multicolumns ``$k$-core'' and ``1-stage'' give the vertex pruning ratio followed by the edge pruning ratio when preprocessed by removing vertices of core number $< \omega - 1$ and our preprocessor, respectively.
For the last three columns, all runtimes are in seconds averaged over three independent runs. The column ``Pruning'' is the time for feature computation \emph{and} pruning. The two remaining columns give the runtime of a solver, containing the runtime on the pruned instance with the speedup obtained in parenthesis. We denote by \texttt{t/o} killed execution after an hour and --- denotes no speedup.}
\def\arraystretch{0.75}
\label{tbl:dense}
\begin{tabular}{lrrrrrrrrrrr}
\toprule
Instance & $|V|$ & $|E|$ & $\omega$ & n.\ $\omega$ & \multicolumn{2}{c}{$k$-core} & \multicolumn{2}{c}{1-stage} & Pruning & \texttt{cliquer} & \texttt{MoMC} \\
\midrule
brock200-1 & 200 & 14.8 K & 21 (20) & 2 (16) & --- & --- & \textbf{0.34} & \textbf{0.55} & $<$0.01 & \textbf{0.39 (53.07)}& 0.04 (44.57) \\
keller4 & 171 & 9.4 K & \textbf{11*} & 2304 (37) & --- & --- & \textbf{0.30} & \textbf{0.50} & $<$0.01 & \textbf{$<$0.01 (38.11)} & 0.02 (5.68) \\
\rowcol keller5 & 776 & 226 K & \textbf{27*} & 1000 (5) & --- & --- & \textbf{0.28} & \textbf{0.48} & 0.19 & \texttt{t/o} & \textbf{1421.24 ($>$2.53)} \\
p-hat300-3 & 300 & 33.4 K & \textbf{36*} & \textbf{10*} & --- & --- & \textbf{0.38} & \textbf{0.58} & 0.02 & \textbf{87.1 (9.12)} & 0.05 (6.00) \\
p-hat500-3 & 500 & 93.8 K & \textbf{50*} & 62 (40) & --- & --- & \textbf{0.34} & \textbf{0.52} & 0.07 & \texttt{t/o} & \textbf{2.51 (5.98)} \\
\rowcol p-hat700-1 & 700 & 61 K & \textbf{11*} & \textbf{2*} & --- & --- & \textbf{0.36} & \textbf{0.47}  & 0.03 & 0.08 (1.22) & \textbf{0.05 (1.30)} \\
p-hat700-2 & 700 & 121.7 K & \textbf{44*} & \textbf{138*} & --- & --- & \textbf{0.36} & \textbf{0.45} & 0.11 & \texttt{t/o} & 1.35 (---)  \\
p-hat1000-1 & 1 K & 122.3 K & \textbf{10*} & 276 (165) & --- & --- & \textbf{0.36} & \textbf{0.47}  & 0.08 & \textbf{0.86 (2.22)} & 0.71 (1.67) \\
\rowcol p-hat1500-1 & 1.5 K & 284.9 K & 12 (11) & 1 (376) & --- & --- & \textbf{0.33} & \textbf{0.43} & 0.25 & 13.18 (---) & \textbf{3.2 (1.54)} \\
fp & 7.5 K & 841 K & \textbf{10*} & \textbf{1001*} & --- & --- & \textbf{0.06} & \textbf{0.29} & 0.36 & 0.65 (---) & \textbf{5.19 (1.13)} \\
nd3k & 9 K & 1.64 M & \textbf{70*} & \textbf{720*} & --- & --- & \textbf{0.23} & \textbf{0.28} & 1.28 & \texttt{t/o} & \textbf{7.05 (1.09)} \\
\rowcol raefsky1 & 3.2 K & 291 K & \textbf{32*} & 613 (362) & --- & --- & \textbf{0.33} & \textbf{0.38} & 0.11 & 2.80 (---) & \textbf{0.31 (1.36)} \\
HFE18\_96\_in & 4 K & 993.3 K & \textbf{20*} & \textbf{2*} & $<$1e-4 & $<$1e-4 & \textbf{0.26} & \textbf{0.27} & 0.49 & 58.88 (1.05) & \textbf{4.30 (1.18)} \\
heart1 & 3.6 K & 1.4 M & \textbf{200*} & 45 (26) & $<$1e-4 & $<$1e-4 & \textbf{0.19} & \textbf{0.25} & 0.66 & \texttt{t/o} & 19.37 (---) \\
\rowcol cegb2802 & 2.8 K & 137.3 K & \textbf{60*} & 101 (38) & 0.09 & 0.04 & \textbf{0.39} & \textbf{0.46} & 0.09 & 0.05 (---) & \textbf{0.15 (1.61)} \\
movielens-1m & 6 K & 1 M & \textbf{31*} & \textbf{147*} & 0.05 & 0.007 & \textbf{0.22} & \textbf{0.23} & 0.98 & 31.31 (---) & \textbf{2.85 (1.14)} \\
ex7 & 1.6 K & 52.9 K & \textbf{18*} & 199 (127) & 0.02 & 0.01 & \textbf{0.26} & \textbf{0.28}  & 0.04 & 0.01 (---) & \textbf{0.1 (1.29)} \\
\rowcol Trec14 & 15.9 K & 2.87 M & \textbf{16*} & \textbf{99*} & 0.16 & 0.009 & \textbf{0.34} & \textbf{0.15} & 2.19 & 3.62 (---) & 0.35 (---) \\
\bottomrule 
\end{tabular}
\end{table*}
\end{landscape}

\paragraph{Classification framework for dense networks}

For training, we get 4762 feature vectors from our ``dense'' category.
As a baseline, a 4-fold cross validation over this using logistic regression results in an accuracy of \textbf{0.73}.
We improve on this by obtaining an accuracy of \textbf{0.81} with gradient boosted trees (further details omitted), found with the help of \texttt{auto-sklearn}~\cite{autosklearn}.

\paragraph{Search strategies}
Given the empirical hardness of dense instances, one should not expect a very high accuracy with polynomial-time computable features such as (F1)-(F10).
For this reason, we set the confidence threshold $q=0.98$ here.

\paragraph{The failure of $k$-core decomposition on dense graphs}
It is common that widely-adopted preprocessing methods like the $k$-core decomposition cannot prune any vertices on a dense network $G$, even if they had the computationally expensive knowledge of $\omega(G)$.
This is so because the degree of each vertex is higher than than the maximum clique size $\omega(G)$.

We showcase precisely this poor behaviour in Table~\ref{tbl:dense}.
For most of the instances, the $k$-core decomposition with the exact knowledge of $\omega(G)$ cannot prune any vertices.
In contrast, the probabilistic preprocessor prunes typically around 30 \% of the vertices and around 40 \% of the edges.

\paragraph{Accuracy}
Given that around 30 \% of the vertices are removed, how many mistakes do we make?
For almost all instances we retain the clique number, i.e., $\omega(G') = \omega(G)$, where $G'$ is the instance obtained by preprocessing $G$ (see column ``$\omega$'' in Table~\ref{tbl:dense}).
In fact, the only exceptions are \mbox{brock200-1} and \mbox{p-hat1500-1}, for which $\omega(G') = \omega(G) - 1$ still holds.
Importantly, for about half of the instances, we retain \emph{all} optimal solutions.

\paragraph{Speedups}
We show speedups for the solvers after executing our pruning strategy in Table~\ref{tbl:dense} (last two columns).
We obtain speedups as large as 53x and for 38x \mbox{brock200-1} and \mbox{keller4}, respectively.
This might not be surprising, since in both cases we lose some maximum cliques (but note that for \mbox{keller4}, the size of a maximum clique is still retained).
For \mbox{p-hat300-3}, the preprocessor makes no mistakes, resulting in speedups of upto 9x.
The speedup for \mbox{keller5} is \emph{at least} 2.5x, since the original instance was not solved within 3600 seconds, but the preprocessed instances was solved in roughly 1421 seconds.

Most speedups are less than 2x, explained by the relative simplicity of instances.
Indeed, it seems challenging to locate dense instances of MCE that are (i) structured and (ii) solvable within a reasonable time.

\subsection{Sparse networks}
\label{subs:sparse}
In this subsection, we show results for probabilistic preprocessing on sparse networks (i.e., edge density below 0.5).

\paragraph{Classification framework for sparse networks} We use logistic regression trained with stochastic gradient descent.
We use a standard L2 regularizer, and use 0.0001 as the regularization term multiplier determined by a systematic grid search.
The classifier is trained for 400 epochs.

\begin{landscape}
\begin{table*}[t]
\centering
\small
\caption{Experiments for sparse graphs. The columns are precisely as in Table~\ref{tbl:dense}, with the exception that we show pruning ratios for~5 stages.
All ratios are rounded to three decimal places. 
Ratios of~1.000 are between~1 and 0.999.
An \texttt{s} marks a segmentation fault.}
\def\arraystretch{0.75}
\label{tbl:pruning}
\begin{tabular}{lrrrrrrrrrrr}
\toprule
Instance & $|V|$ & $|E|$ & $\omega$ & n.~$\omega$ & \multicolumn{2}{c}{$k$-core} & \multicolumn{2}{c}{5-stage} & Pruning & \texttt{cliquer} & \texttt{MoMC} \\
\midrule
bio-WormNet-v3 & 16 K & 763 K &  \textbf{121*} & \textbf{18*} & 0.868 & 0.602 & \textbf{0.987} & \textbf{0.975} & 0.36 & 0.37 (---) & \textbf{0.40 (3.94)} \\
ia-wiki-user-edits-page & 2 M & 9 M & \textbf{15*} & \textbf{15*} & 0.958 & 0.641 & \textbf{0.997}	& \textbf{0.946} & 1.12 & \textbf{1.16 (29.94)} & \texttt{s} \\
\rowcol rt-retweet-crawl & 1 M & 2 M &   \textbf{13*} &   \textbf{26*} &      0.979 & 0.863 & \textbf{0.997} & \textbf{0.989} & 0.38 & \textbf{0.41 (5.66)} & \texttt{s} \\
soc-digg	& 771 K & 6 M & \textbf{50*} & \textbf{192*} & 0.969 & 0.496 & \textbf{0.998}	& \textbf{0.964} & 4.80 & \textbf{4.91 (1.78)} & \texttt{s} \\
soc-flixster & 3 M & 8 M &    \textbf{31*} &  \textbf{752*}  &      0.986 & 0.834 & \textbf{0.999} & \textbf{0.989} & 1.32 & \textbf{1.41 (3.86)} & \texttt{s} \\
\rowcol soc-google-plus & 211 K & 2 M &   \textbf{66*} &   \textbf{24*} &      0.986 & 0.785 & \textbf{0.998} & \textbf{0.972} & 0.35 & 0.35 (---) & \textbf{0.41 (3.98)} \\
soc-lastfm & 1 M & 5 M &    \textbf{14*} &  330 (324) &      0.933 & 0.625 & \textbf{0.993} & \textbf{0.938} & 2.24 & \textbf{2.57 (10.56)} & \texttt{s} \\
soc-pokec & 2 M & 22 M & \textbf{29*} & \textbf{6*} & 0.824 & 0.595 & \textbf{0.975} & \textbf{0.940} & 17.59 & \textbf{24.40 (45.80)} & \texttt{s} \\
\rowcol      soc-themarker & 69 K & 2 M &   \textbf{22*} &   \textbf{40*} &      0.713 & 0.151 & \textbf{0.972}	& \textbf{0.842} & 2.03 & \textbf{4.95 (---)} & \texttt{s} \\
soc-twitter-higgs	& 457 K & 15 M & \textbf{71*} & \textbf{14*} & 0.852 & 0.540 & \textbf{0.986}	& \textbf{0.943} & 9.52 & \textbf{9.85 (1.92)} & \texttt{s} \\
 soc-wiki-Talk-dir & 2 M & 5 M &   \textbf{26*} &  \textbf{141*} &      0.993 & 0.830 & \textbf{0.999} & \textbf{0.970} & 1.09 & \textbf{3.47 (1.25)} & \texttt{s} \\
\rowcol      socfb-A-anon & 3 M & 24 M &   \textbf{25*} &   \textbf{35*} &      0.879 & 0.403 & \textbf{0.984} & \textbf{0.907} & 28.49 & \textbf{38.05 (55.95)} & \texttt{s} \\
      socfb-B-anon & 3 M & 21 M &   \textbf{24*} &  \textbf{196*} &     0.884 & 0.378 & \textbf{0.986} & \textbf{0.920} & 28.33 & \textbf{35.49 (67.46)} & \texttt{s} \\
     socfb-Texas84 & 36 K & 2 M &  \textbf{51*} &   \textbf{34*} &      0.540 & 0.322 & \textbf{0.957} & \textbf{0.941} & 1.04 & \textbf{1.07 (1.32)} & \texttt{s} \\
\rowcol tech-as-skitter & 2 M & 11 M &   \textbf{67*} &    \textbf{4*}  &      0.997 & 0.971 & \textbf{1.000} & \textbf{0.998} & 0.28 & 0.28 (---) & \textbf{0.36 (4.31)} \\
   web-baidu-baike & 2 M & 18 M & \textbf{31*} & \textbf{4*} & 0.933 & 0.618 & \textbf{0.992} & \textbf{0.934} & 9.67 & \textbf{11.00 (7.48)} & \texttt{s} \\
    web-google-dir & 876 K & 5 M &   \textbf{44*} &    \textbf{8*} &      1.000 & 0.999 & \textbf{1.000} & \textbf{1.000} & $<$ 0.00 & $<$ 0.00 (---) & \textbf{$<$ 0.00 (2.06)} \\
\rowcol        web-hudong & 2 M & 15 M &  267 (266) &   59 (1) &      1.000 & 0.996 & \textbf{1.000} & \textbf{0.997} & 0.09 & 0.10 (---) & \textbf{0.1 (9.99)} \\
 web-wikipedia2009 & 2 M & 5 M &   \textbf{31*} &    \textbf{3*} & 0.999 & 0.988 & \textbf{1.000} & \textbf{1.000} & 0.03 & 0.03 (---) & \textbf{0.03 (4.28)} \\
\bottomrule
\end{tabular}
\end{table*}
\end{landscape}

\paragraph{Implementing the $k$-core decomposition}
Recall the exact state-of-the-art preprocessor: (i) use a heuristic to find a large clique (say of size $k$) and (ii) delete every vertex of $G$ of core number less than $k-1$.
For sparse graphs, a state-of-the-art solver \texttt{pmc} has been reported to find large cliques, i.e., typically $k$ is at most a small additive constant away from $\omega(G)$\footnote{A table of results seen at \url{http://ryanrossi.com/pmc/download.php}}.
Further, given that some real-world sparse networks are scale-free (many vertices have low degree) the $k$-core decomposition can be effective in practice.

To ensure highest possible prune ratios for the $k$-core decomposition method, we supply it with the number $\omega(G)$ instead of an estimate provided by any real-world implementation.
This ensures \emph{ideal conditions}: (i) the method always prunes as aggressively as possible, and (ii) we further assume its execution has zero cost. We call this method the \emph{$\omega$-oracle}.

\paragraph{Test instance pruning}
Before applying our preprocessor on the sparse test instances, we prune them using the $\omega$-oracle.
This ensures that the pruning we report is highly non-trivial, while also speeding up feature computation.

\paragraph{Search strategies}
We experiment with the following two multi-stage search strategies:
\begin{itemize}
\item \emph{Constant confidence (CC):} at every stage, perform probabilistic preprocessing with confidence threshold $q$.
\item \emph{Increasing confidence (IC):} at the first stage, perform probabilistic preprocessing with confidence threshold $q$, progressing $q$ by $d$ for every later stage.
\end{itemize}
Our goal is two-fold: to find (i) a number of stages $\ell$ and (ii) parameters $q$ and $d$, such that the strategy never errs while pruning as aggressively as possible.
We do a systematic search over parameters $\ell$, $q$, and $d$. 
For the CC strategy, we let $\ell \in \{1,2,\ldots,8\}$ and $q \in \{ 0.55, 0.6, \ldots, 0.95 \}$.
For the IC strategy, we try $q \in \{ 0.55, 0.60, 0.65 \}$, $d = 0.05$, and set $\ell$ so that in the last stage the confidence is 0.95.

We find the CC strategy with $q = 0.95$ to prune the highest while still retaining all optimal solutions.
Thus, for the remaining experiments, we use a CC strategy with $q=0.95$.

Our 5-stage strategy outperforms, almost always safely, the $\omega$-oracle (see Table~\ref{tbl:pruning}).
In particular, note that even if the difference between the vertex pruning ratios is small, the impact for the number of edges removed can be considerable (see e.g., all instances of the ``soc'' category).
We note that the runtime is not sensitive to the number of stages $\ell$. In fact, already the first step of pruning makes the graph so small that further stages add comparatively very small amounts to the overall runtime.


\paragraph{Speedups}
We show speedups for the solvers in Table~\ref{tbl:pruning}.
We use as a baseline the solver executed on an instance pruned by the $\omega$-oracle, which renders many of the instances easy already.
Most notably, this is \emph{not} the case for \mbox{soc-pokec}, \mbox{socfb-A-anon}, and \mbox{socfb-B-anon}, all requiring at least 5 minutes of solver time.
The largest speedup is for \mbox{socfb-B-anon}, where we go from requiring 40 minutes to only 7 seconds of solver time.
For \texttt{MoMC}, most instances report a segmentation fault (marked with an \texttt{s}) for an unknown reason.


\subsection{Edge-based classification}
For edges, we do a similar training as that described for vertices.
For the category ``dense'', we obtain 79472 feature vectors.
Further, for this category, the edge classification accuracy is \textbf{0.83}, which is 1~\% higher than the vertex classification accuracy using the same classifier as in Subsection~\ref{subs:dense}. However, we note that the edge feature computation is noticeably slower than that for vertex features.
This reason combined with the fact that the classification accuracy is almost the same, we omit further experiments with the edge features due to smaller speedups.

\subsection{Model analysis}
\label{subs:analysis}

\begin{figure*}[t]
\centering
	\begin{tabular}{cc}
	\hspace*{0mm}	\includegraphics[width=0.5\columnwidth,keepaspectratio]{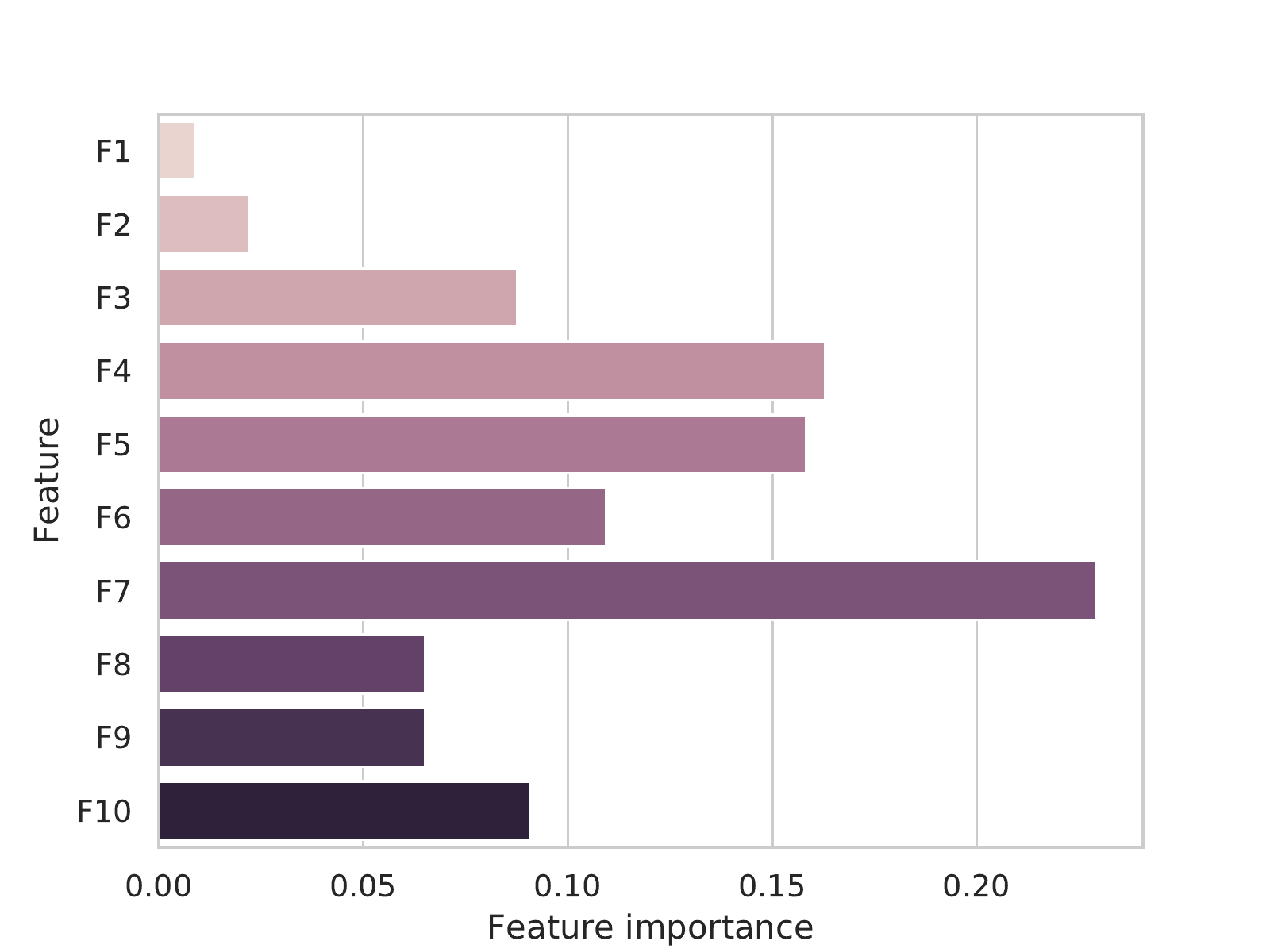} &
	\hspace*{0mm}	\includegraphics[width=0.5\columnwidth,keepaspectratio]{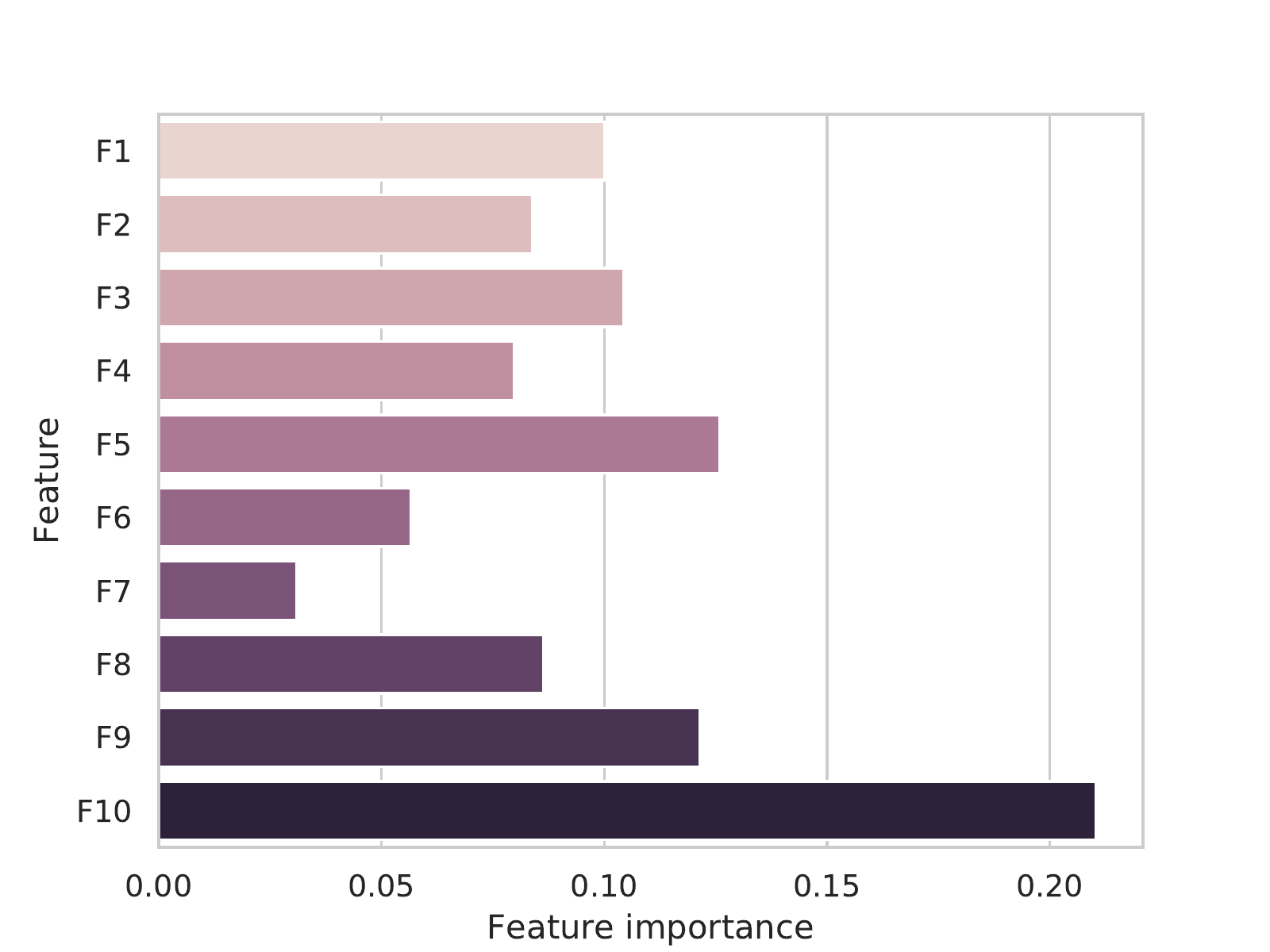} \\
	(a) Dense networks & (b) Sparse networks 
	\end{tabular}
	\caption{The feature importance for \textbf{(a)} dense nets (category ``dense'') and \textbf{(b)} sparse nets (category ``all'').}
	\label{fig:featimp}
\end{figure*}

Gradient boosted trees (used with dense networks in Subsection~\ref{subs:dense}) naturally output feature importances.
We apply the same classifier for the sparse case to allow for a comparison of feature importance.
In both cases, the importance values are distributed among the ten features and sum up to one.

Unsurprisingly, for sparse networks, the local chromatic density (F10) dominates (importance 0.22).
In contrast, (F10) is ineffective for dense networks (importance 0.08), since the chromatic number tends to be much higher than the maximum clique size.
In both cases, (F5) eigencentrality has relatively high importance, justifying its expensive computation.

For dense networks, (F7) average $\chi^2$ over neighbor degrees has the highest importance (importance 0.23), whereas in the sparse case it is least important feature (importance 0.03).
This is so because all degrees in a dense graph are high and the degree distribution tends to be tightly bound or coupled.
Hence, even slight deviations from the expected (e.g., vertices in large cliques) depict high statistical significance scores.
We will capitalize on this observation later on in Section~\ref{sec:althea}.

\section{On supervised learning for hard problems}
\label{sec:planted}
The goal of this section is two-fold: 
(i) to explain the high accuracy of our proposed framework, even when it was trained with small instances, and
(ii) as a consequence, argue that supervised learning is a viable approach for solving structured instances of certain hard problems.

To ensure that the input instances are, at some point, ``structure-free'' we turn to the following heavily-studied variant of the maximum clique problem.
This serves as a representative of the \emph{worst-case input} for our preprocessing strategy.
Also, observe that in case the input graph has a unique maximum clique, MCE is equivalent to finding the (single) maximum clique.
For simplicity, we restrict ourselves to single stage sparsification in these experiments.

\subsection{Planted clique}
\label{subs:plantedc}
In the \emph{planted clique problem}~\cite{Jerrum1992,Kucera1995}, we are given an Erd\H{o}s-R\'{e}nyi random graph $H := G(n,p)$, i.e., an $n$-vertex graph where the presence of each edge is determined independently with probability $p$ (see~\cite{Erdos1959}).
In addition, the problem is parameterized by an integer $k$ such that a random subset of $k$ vertices has been chosen from $H$ and a clique added on it.
On this input, the task is to identify (with the knowledge of the value of $k$) the $k$ vertices containing the planted clique.

The problem is easy for $k \leq \log_2(n)$.
In particular, as shown in~\cite{Bollobas2013}, the clique number of $G(n,p)$ as $n \to \infty$ is almost surely $w$ or $w+1$ where $w$ is the greatest natural number such that
\begin{equation}
\label{eq:clnum}
{n \choose w} p^{w \choose 2} \geq \log(n),
\end{equation}
where $w$ is roughly $2 \log_2 (n)$.
Even when a clique of such size is known to exist (whp), we only know how to find a clique of size $\log_2(n)$ efficiently,\footnote{It is conjectured~\cite{Karp1976,Feldman2017} that there is no polynomial-time algorithm for finding a clique of size $(1+\eps) \log_2 n$ for any $\eps > 0$ in $G(n,1/2)$.} and also solve the problem in polynomial-time when $k$ is large enough.
Specifically, it is known that several algorithmic techniques such as spectral methods (see e.g.,~\cite{Feldman2017} for more) produce efficient algorithms for the problem when $k = \Omega(\sqrt{n})$.

However, settling the complexity of the problem is a notorious open problem when $k$ is between $2 \log_2(n)$ and $\sqrt{n}$. Next, we will focus precisely on this difficult region.

\begin{figure}[t]
\centering
	\begin{tabular}{ccc}
	\hspace*{-4mm}	\includegraphics[width=0.33\columnwidth]{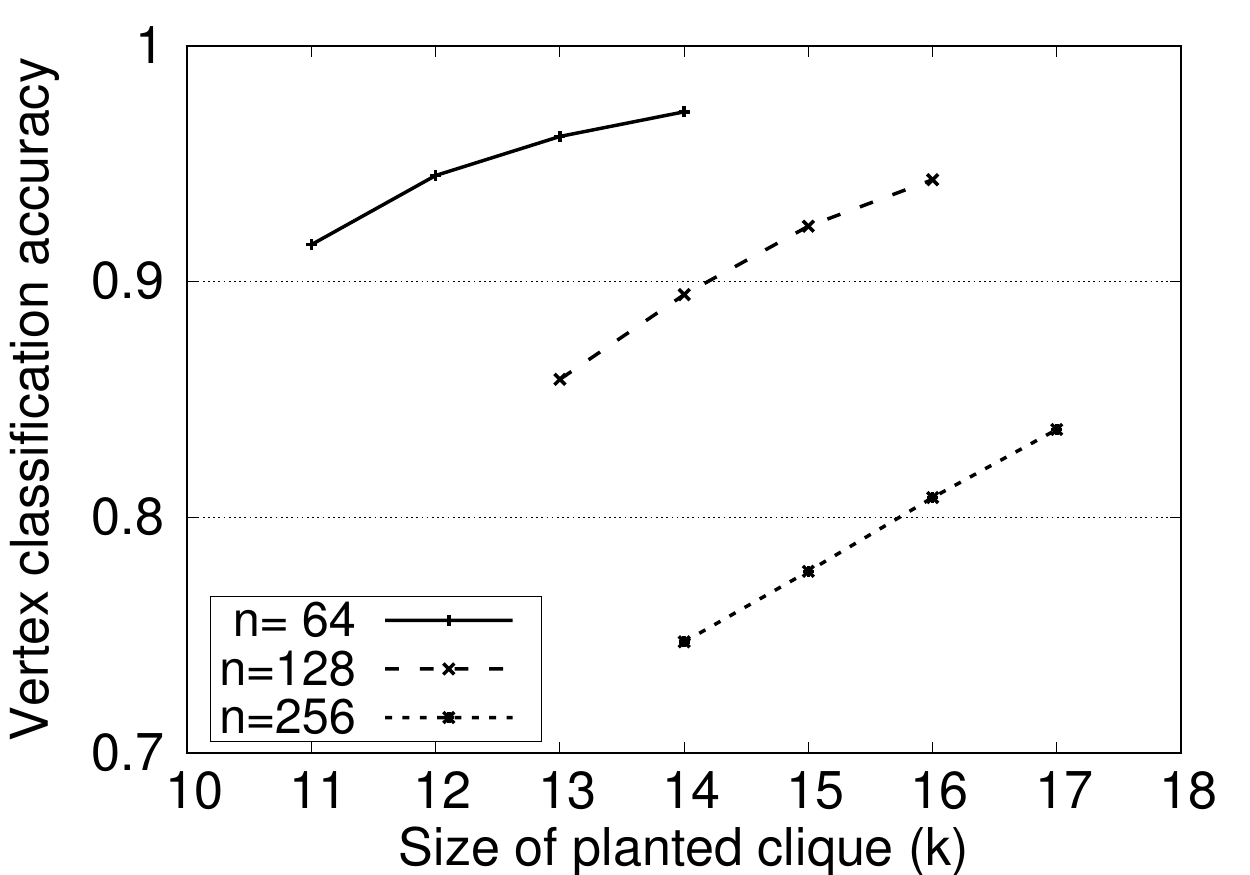} &
	\hspace*{-6mm}	\includegraphics[width=0.33\columnwidth]{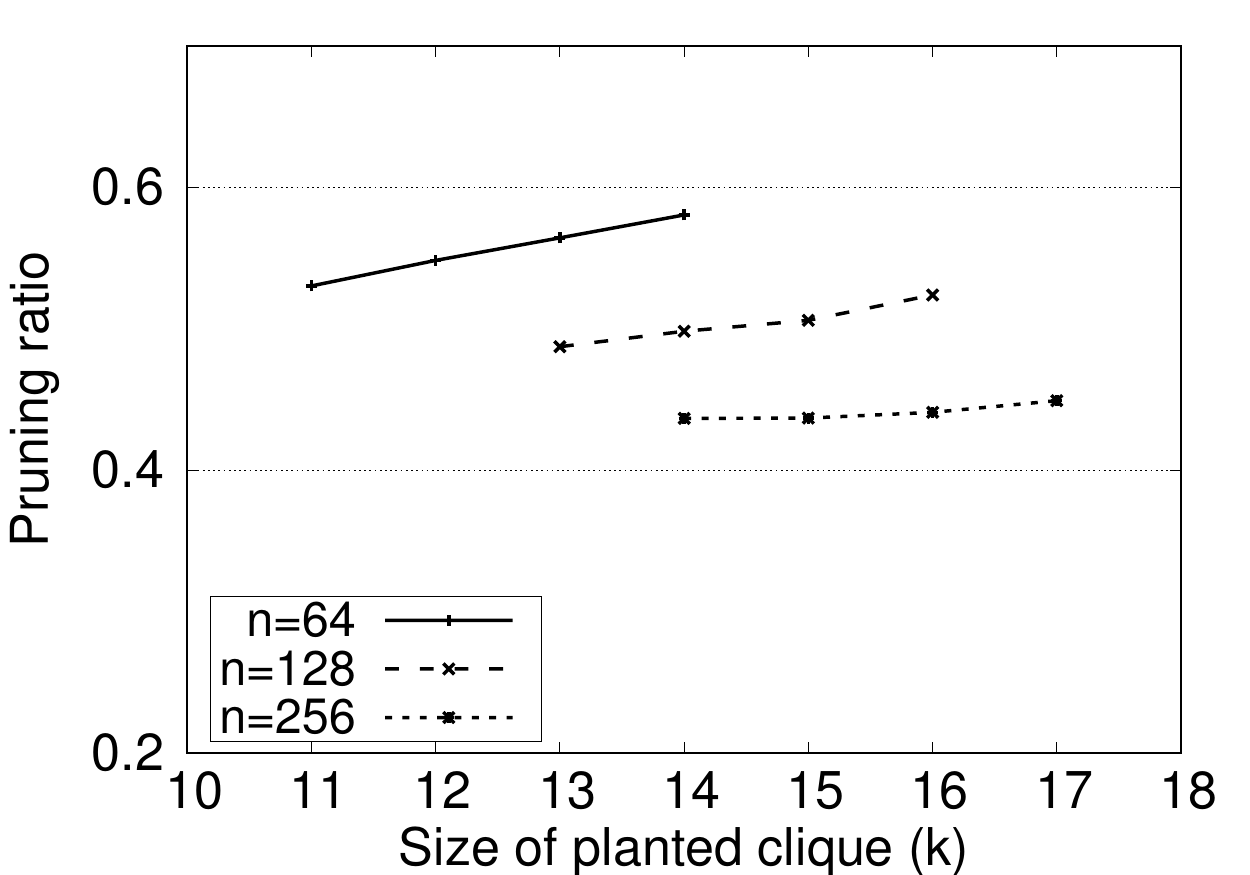} &
	\hspace*{-6mm}	\includegraphics[width=0.33\columnwidth]{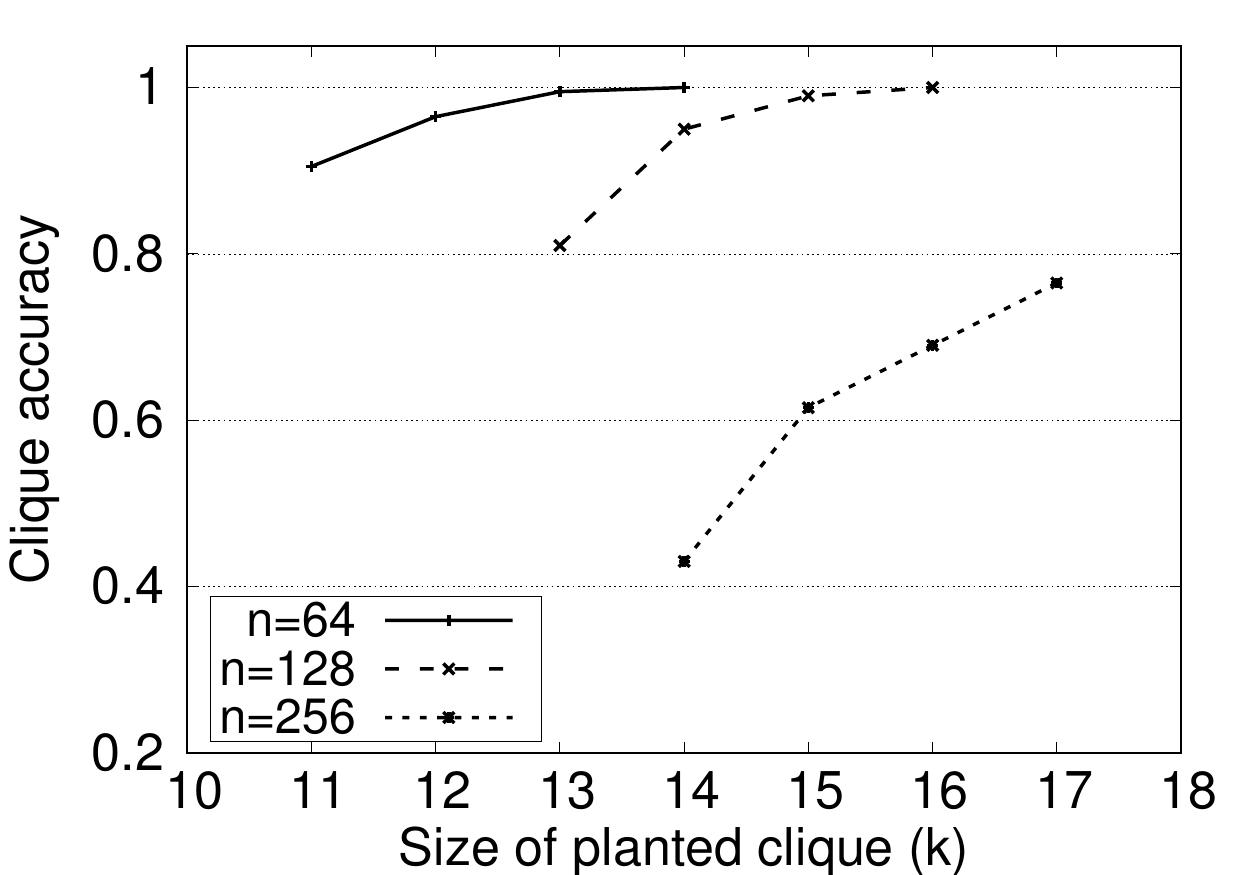} \\
	(a) Vertex acc. & (b) Pruning ratio & (c) Clique acc.
	\end{tabular}
	\caption{The vertex accuracy, pruning ratio, and clique accuracy of our framework when trained with $G(n,1/2)$ with three different parameter pairs $(64,10)$, $(128,12)$, and $(256,13)$. The predictions are for independent, distinct samples with the same $n$, but growing planted clique size~$k$.}
	\label{fig:acc}
\end{figure}

\subsection{Pushing the limits of preprocessing}
\label{subs:limits}
In this subsection, we explore the limits of scalability and robustness of our framework on the planted clique problem.
All experiments are done on an Intel Core i5-6300U CPU (2.4 GHz), 8 GB of RAM, running Ubuntu 16.04, differing only slightly from the earlier hardware configuration.
For all experiments here, we use only the \texttt{igraph} algorithm.

\paragraph{Generation of synthetic data}
We use the \texttt{genrang} utility program~\cite{McKay2014} to sample a random graph~$H := G(n,p)$.
To plant a clique of size $k$, we sample uniformly at random~$k$ vertices, denoted by $K$, from~$H$ and insert all corresponding at most ${k \choose 2}$ missing edges into~$H$.

For each $H$, we compute the features described in Section~\ref{sect:features} with the following differences: we replace (F10) the local chromatic density with the order-four LCC and modify (F8) and (F9) to consider order-four LCC instead of the LCC.
This brings more predictive power while still remaining computationally feasible for small graphs.
The values $E_i$ in Equation~\ref{eq:chis} for (F6) and (F7) are the expected degree $n \cdot p$, while for (F8) and (F9) they are the expected order-$k$ LCC given as ${n - 1 \choose k-1} p^{k \choose 2} \cdot {1} / {np \choose k-1}$. 
To ensure a balanced dataset, we sample (i) $k$ label-0 examples from $K$ and (ii) $k$ label-1 examples from $G \setminus K$, both uniformly at random.

For training, we consider $n \in \{64,128,256,512\}$ because the clique number grows roughly logarithmically with $n$ (see Equation~\ref{eq:clnum}).
We fix $p = 1/2$.
For every $n$, we compute $w$ from Equation~\ref{eq:clnum}, and sample graphs $G(n,p)$ with a planted clique of size $k \in \{w+2,\ldots,w+6\}$ such that each pair $(n,k)$ gives a dataset of size at least 100~K feature vectors.
When planting a clique of size at least $w+2$, we try to guarantee the existence of a unique maximum clique in the graph.
However, this procedure does not always succeed due to randomness, but we do not discard such rare outcomes. 

\begin{figure}[t]
\centering
	\begin{tabular}{ccc}
	\includegraphics[width=0.28\columnwidth]{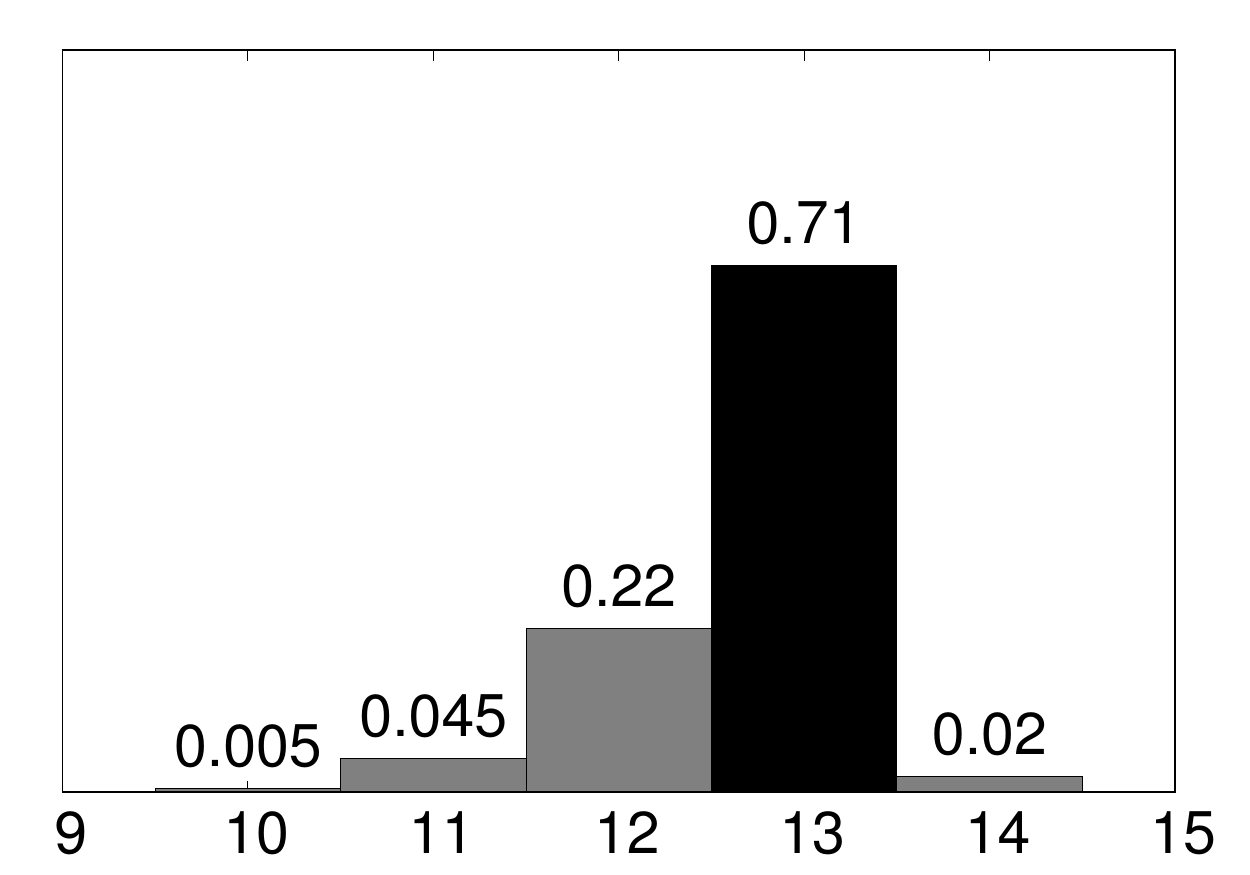} &
	\includegraphics[width=0.28\columnwidth]{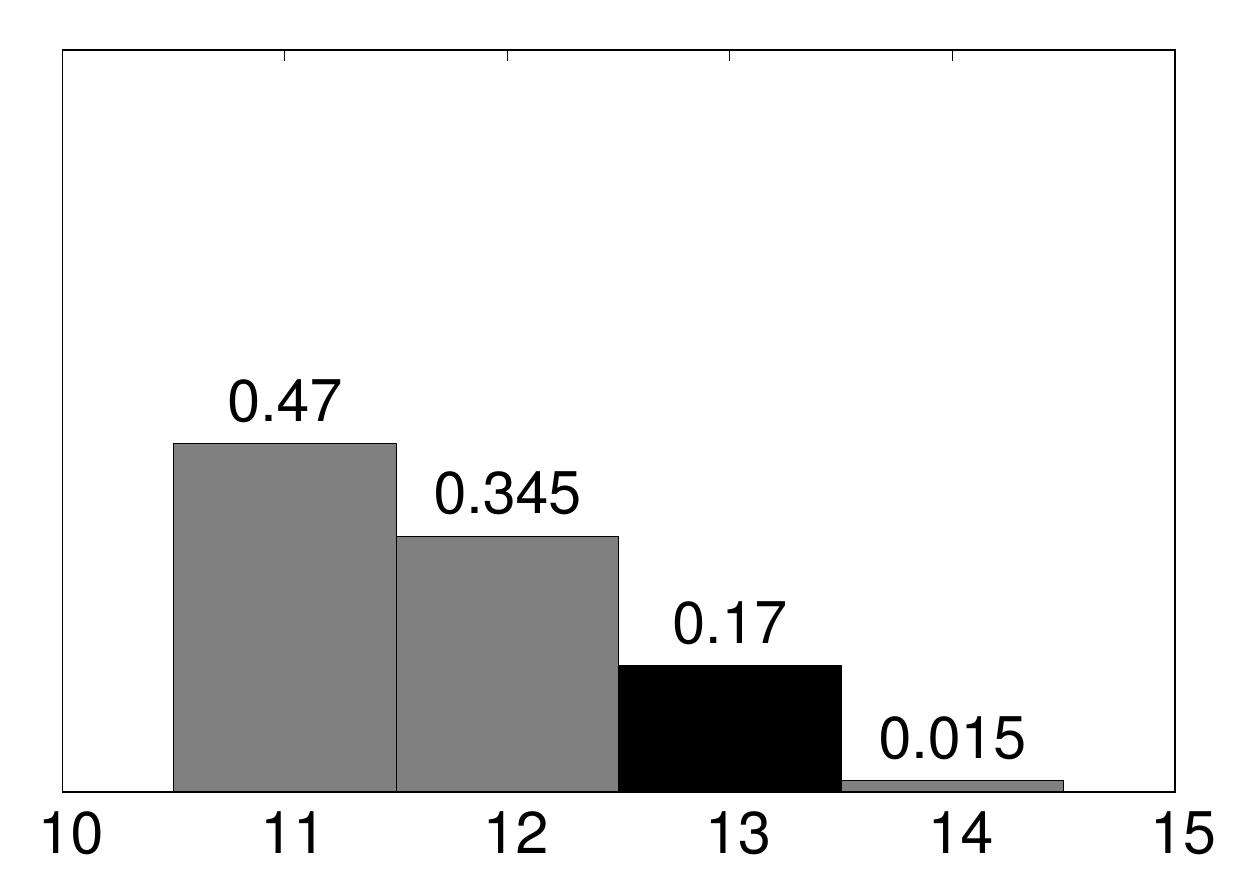} &
	\includegraphics[width=0.28\columnwidth]{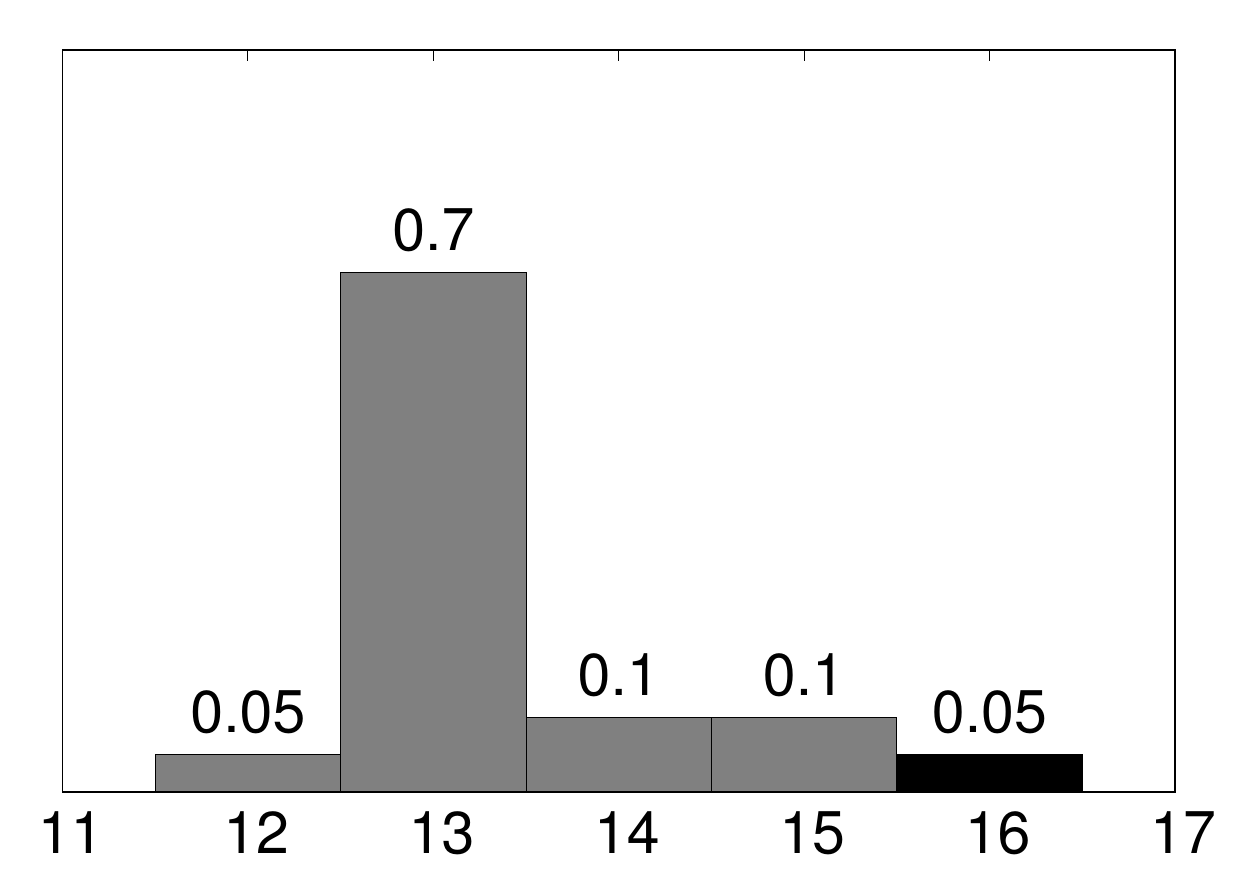} \\
	(a) $(128,13)$ & (b) $(256, 13)$ & (c) $(512, 16)$
	\end{tabular}
	\caption{Distribution of extracted maximum clique size, with black bars denoting the size of the planted clique. Both (a) and (b) are over 200 samples, while (c) is over 20 samples. In each, the predicting classifier has been trained with 64-vertex random graphs with a planted clique of size 10.}
	\label{fig:cldist}
\end{figure}

\begin{table*}
  \centering
  \small
  \caption{Robustness and speedups with fixed $n$ and increasing $k$. The leftmost two columns show the data $(n,k)$ used to train a classifier $P$. For each planted clique size $k+1$, $k+2$, and $k+3$, we show the average pruning ratio (column ``Pr.''), the average clique accuracy (column ``Acc.'), the average runtime of \texttt{igraph} on the reduced instance obtained from our framework using $P$ (column ``Time~(s)''), and the average speedup over executing the same algorithm on the original instance.}
  \label{tbl:robust}
  \def\arraystretch{1}
  \setlength{\tabcolsep}{3pt}
  \begin{tabular}{*{111}{l}}
    \toprule
    & & & \multicolumn{2}{c}{$k+1$} & & & \multicolumn{2}{c}{$k+2$} & & & \multicolumn{2}{c}{$k+3$} \\
    \cmidrule(lr){3-6}
    \cmidrule(lr){7-10}
    \cmidrule(lr){11-14}
    $n$ & $k$ & Pr. & Acc. & Time & Speedup & Pr. & Acc. & Time & Speedup & Pr. & Acc. & Time & Speedup \\
    \midrule
    \phantom{0}64 & 10 & 0.530 & 0.905 & 0.068 & 0.132 & 0.548 & 0.965 & 0.068 & 0.135 & 0.564 & 0.995 & 0.068 & 0.135 \\
    \rowcol 128 & 12 & 0.506 & 0.710 & 0.301 & 0.759 & 0.517 & 0.875 & 0.296 & 0.774 & 0.525 & 0.935 & 0.297 & 0.784 \\
    256 & 13 & 0.489 & 0.170 & 3.261 & 3.264 & 0.493 & 0.190 & 3.233 & 3.304 & 0.493 & 0.310 & 3.260 & 3.315 \\
    \rowcol 512 & 15 & 0.492 & 0.05 & 70.587 & 12.994 & 0.492 & 0.05 & 70.086 & 12.816 & 0.491 & 0.100 & 70.562 & 12.722 \\
    \bottomrule
  \end{tabular}
\end{table*}

\paragraph{Vertex classification accuracy}
We study the accuracy of our classifiers for distinguishing vertices that are and are not in a maximum clique.
Specifically, we train a classifier for each pair $(n,k) \in \{ (64,10), (128, 12), (256, 12) \}$, and test for unseen graphs with the same $n$ but growing planted clique size $k' = k + 1, \ldots, k+4$.
The results are shown in Figure~\ref{fig:acc}~(a).
As expected, the classification task becomes easier once $k'$ increases.
This is also supported the fact that multiple algorithms solve the planted clique problem in polynomial-time for large enough $k'$ (see Subsection~\ref{subs:plantedc}).
In addition, as $n$ grows larger, we see accuracy deterioration caused by the converge of the local properties towards their expected values.
Especially for small values of $k'$, the injection of the planted clique is not substantial enough to cause significant deviations from the expected values.

\paragraph{Pruning ratio and clique accuracy}
We study the effectiveness of our framework as a probabilistic preprocessor for the planted clique instances.
We fix the confidence threshold $q = 0.55$ and use the same set of classifiers and test data.
The average pruning ratios over all instances are shown in Figure~\ref{fig:acc}~(b).
We see pruning ratios as high as at most~0.6, while always discarding more than 40~\% of the vertices.

Now, it is possible that $P$ makes an erroneous prediction causing the deletion of a vertex, which in turn lowers the size of a maximum clique in the instance.
The average clique accuracies over all instances are shown in Figure~\ref{fig:acc}~(c).
Here, we see that for $n=256$, the vertex accuracy (Figure~\ref{fig:acc}~(a)) is still above 0.7, but the clique accuracy drops to above 0.4.
As the vertex accuracy decreases, the probability of deleting a vertex present in a maximum clique increases, translating to a higher chance of error in extracting a maximum clique.
However, while not completely error-free, we observe that even in the case of $(256,13)$ we always delete at most two members of a maximum clique, whereas in the case of $(512,16)$, 95~\% of the time, we extract a maximum clique of size at least~13 (see Figure~\ref{fig:cldist}).

\paragraph{Robustness and speedups}
The robustness and speedups obtained using the \texttt{igraph} algorithm are given in Table~\ref{tbl:robust}.
Here, the clique accuracy and runtime are obtained as the average over 200 samples for each $(n,k)$ except for $(n,k) = (512,15)$ for which there are 20 independent samples.
We see the drop in clique accuracy when a classifier $P$ is trained with $(n,k) \in \{ (256,12), (512,15) \}$ and is predicting for the same $n$ but increasing~$k$.
The clique accuracy is a strict measure, so to quantify the severeness of the erroneous predictions made by $P$ we show the distributions of the extracted maximum clique sizes in Figure~\ref{fig:cldist} for some pairs $(n,k)$.
Again, we observe the effects of growing $n$ causing the convergence of local properties, consequently decreasing the predictive power of $P$.
For $(n,k) = (128,13)$, 73 \% of the runs still produce an optimal solution (here, one can also observe the rare event of having a maximum clique of size~14 when the planted clique was of size~13).

\paragraph{The case for supervised learning on intractable problems}

\begin{table}[t]
  \centering
  \small
  \caption{Deviation in vertex classification accuracy.}
  \label{tbl:retrain}
  \def\arraystretch{1}
  \setlength{\tabcolsep}{4pt}
  \begin{tabular}{llcc}
    \toprule
    $n$ & $k$ & Trained acc.\ & Rob.\ acc.\ \\
    \midrule
    128 & 12 & 0.858 & 0.844 \\
    \rowcol 256 & 13 & 0.747 & 0.728 \\
    512 & 15 & 0.678 & 0.665 \\
    \bottomrule
  \end{tabular}
\end{table}

As $n$ grows, the instances get increasingly time-consuming to solve even for state-of-the-art solvers for suitable $k$, as there is no exploitable structure.
Consequently, obtaining optimally labeled data becomes practically impossible for large enough $n$.
However, in our experiments, we find that random graphs with $n=64$ and $k=10$ are representative of the input for moderately larger graphs as well, up to a point.
Further, obtaining the optimal label for such small graphs is fast.

{We show the deviation in vertex classification accuracy in Table~\ref{tbl:retrain}.}
The column ``Trained acc.'' corresponds to the accuracy of the classifier trained with the values $n$ and $k$ mentioned in the two first columns, while the column ``Rob.\ acc.'' is the accuracy of a classifier trained with smaller $(n,k) = (64,10)$ instances, and predictions are made for the specified $(n,k)$ with planted clique size $k + 1$.
A key observation is that the difference between the two accuracies in a single row in Table~\ref{tbl:retrain} is small enough not to warrant training on larger instances.
This offers an explanation for the perfect clique accuracy with limited training, observed earlier for sparse real-world networks.
This observation reduces the need of labeling costly data points for training.

\section{ALTHEA: a novel clique-finding heuristic for dense graphs}
\label{sec:althea}
In this section, we capitalize on the observation we made in Subsection~\ref{subs:analysis}.
In particular, we describe a heuristic we call {\em ALTHEA} for extracting an approximate maximum clique from a \emph{simple} input graph $G=(V,E)$.

\subsection{Description of ALTHEA}
{\em ALTHEA} hinges on categorizing the {\em degree} of each vertex in $G$ based on its deviation from the average degree of $G$. 
Each vertex is subsequently represented by a sequence of category symbols encoding its neighbourhood, which are then used for computing its statistical significance score. Any vertex depicting the 
maximum $\chi^2$ value (along with its neighbourhood) forms a candidate region for containing a maximum clique in $G$. 
{\em ALTHEA} comprises the following 5 steps.

\paragraph{1. Initialization} We compute the following three degree characteristics of $G$.
\begin{itemize}
\item (i) ${\it \Delta(G)}$: the maximum degree of any vertex in $G$,
\item (ii) ${\it \mathsf{a}(G)}$: the average degree of the vertices in $G$; and
\item (iii) ${\it \sigma(G)}$: the standard deviation of the vertex degrees of $G$.
\end{itemize}
Formally, we define
\begin{equation}
\Delta(G) = \max \{\deg(v): \forall v \in V\}, 
\end{equation} 
\begin{equation}
\mathsf{a}(G) = \frac{\sum_{\forall v \in V} \deg(v)}{|V|},
\end{equation}
and
\begin{equation}
\sigma(G) = \sqrt{\frac{\sum_{\forall v \in V} (\deg(v) - \mathsf{a}(G))^2}{|V| - 1}},
\end{equation}
where $\deg(v)$ is degree of vertex $v$.

\paragraph{2. Symbol categorization} {\em ALTHEA} captures the nature of vertex degree deviation (in the number of standard deviations) 
from the underlying degree distribution of $G$. The number of {\em category symbols} $\tau^G$ is 
$\lceil (\Delta(G) - \mathsf{a}(G))/\sigma(G) \rceil + 1$.
The obtained set of category symbols is $\Gamma^G = \{\gamma_1, \gamma_2, \ldots, \gamma_{\tau^G}\}$, where $\gamma_i$ 
is the multiple of $\sigma(G)$ by which the degree of $v_i$ deviates from $\mathsf{a}(G)$.
Next, we compute the expected probability of occurrence for the symbols in $\Gamma^G$. 

To this end, we use {\em Chebyshev's inequality}~\cite{cheby}, 
which for a random variable $X$ and a real number $k > 0$ states that $\Pr(|X-\mu| \geq k\delta) \leq 1/k^2$, where $\mu$ and $\delta$ are the mean and standard deviation, respectively, 
of the distribution from which $X$ is drawn. Thus, the occurrence probability of $\gamma_i$ is given by 
$\Pr(\gamma_i) = 1/i^2 - 1 / (i+1)^2$.

Other tail distribution bounds or domain-dependent probability distributions 
capturing the underlying characteristics of $G$ might also be used depending on the application. This makes {\em ALTHEA} robust to diverse domains, 
applicable to different input distributions.

\paragraph{3. Vertex symbol sequence} 
For each vertex $v \in V$, we extract its \emph{closed neighbourhood} $N[v] = \{v\} \cup \{u : (u,v) \in E\}$. 
The vertex $v$ is then represented 
by a {\em sequence of category symbols} $\Seq(v)$ of length $|N[v]|$ based on the symbol categorization of the degree of the vertices in its neighbourhood $N[v]$. Formally, this is given as
\begin{equation}
\Seq(v) = \{\gamma(u) : \forall u \in N[v]\},
\end{equation}
where $\gamma(u)$ is the unique $\gamma_i \in \Gamma^G$ and $i \in \{1,2, \ldots, \tau^G\}$, for which the inequality 
$i \leq (\deg(u) - \mathsf{a}(G))/\sigma(G) + 1 < i+1$ holds.

\paragraph{4. Statistical significance computation} For each vertex $v$, {\em ALTHEA} computes the $\chi^2$ statistical significance score using $\Seq(v)$ and the associated symbol probabilities.
For each category symbol $\gamma_i \in \Gamma^G$, its expected occurrence count for vertex $v$ is computed 
as $E_{\gamma_i}^v = \Pr(\gamma_i) \times |N[v]|$. 
Similarly, the corresponding {\em observed occurrence count} $O_{\gamma_i}^v$ of the category symbol $\gamma_i$ for $v$ can be obtained from $\Seq(v)$.
Combining the above steps, the statistical significance of $v$ is 
\begin{equation}
\chi^2(v) = \sum_{\forall \gamma_i \in \Gamma^G} \frac{(O_{\gamma_i}^v - E_{\gamma_i}^v)^2}{E_{\gamma_i}^v}.
\end{equation}

\paragraph{5. Approximate maximum clique extraction} After computing the statistical significance of the vertices, {\em ALTHEA}
selects the vertex $v'$ demonstrating the maximum statistical significance (chosen arbitrarily in case of ties), as the best candidate whose neighbourhood contains an 
(approximate) maximum clique for $G$. Intuitively, a vertex and its neighbours that are a part of a maximum clique in $G$ would exhibit the largest variation in the 
degree distribution characteristic compared to the average (or expected) characteristic of $G$, which is captured by the notion of statistical significance.
Finally, the subgraph induced by the neighbourhood $N[v']$ is fed to a maximum clique solver for extracting a large clique of $G$.

\paragraph{Discussion} In a {\em dense graph} the degree of a vertex is high, and the degree distribution tends to be tightly bound (or coupled). Hence, 
even slight deviations from the expected behaviour (in cases of vertices that are a part of large cliques) depict high statistical significance scores. This enables 
{\em ALTHEA} to effectively identify large maximum cliques, as we will experimentally show next.

\begin{landscape}
\begin{table*}[t]
\centering
\small
\caption{Performance comparison of {\em ALTHEA} on real-world datasets.
Here, (i) $\omega$ denotes the maximum clique size and $\omega'$ is the approximate maximum clique size found; (ii) results for approaches that were ``killed'' after~5 minutes of 
run-time (without output) are marked with $-$; (iii) for results marked with $\#$, refer to Table~\ref{tab:real_large} for additional results; (iv) averaged run-times over $5$ runs 
are shown in seconds; and (v) vertex and edge pruning (Pr.) are given in percentage of $|V|$ and $|E|$ respectively.}
\label{tab:real_small}
\def\arraystretch{1}
\setlength{\tabcolsep}{3pt}
	\begin{tabular}{lllllllllllllll}
	\toprule
	\multicolumn{4}{c}{\bf Dataset Characteristics} & \multicolumn{8}{c}{\bf Heuristic Approaches} & \multicolumn{3}{c}{\bf Exact Approaches} \\
	\cmidrule(lr){1-4}
	\cmidrule(lr){5-12}
	\cmidrule(lr){13-15}
	{\em Instance} & {\em $|V|$} & {\em $|E|$} & {\em $\omega$} & \multicolumn{4}{c}{\bf ALTHEA + FMC(H)} & \multicolumn{2}{c}{\bf FMC(H)} & \multicolumn{2}{c}{\bf RMC} & {\bf FMC(E)} & {\bf MoMC} & {\bf igraph} \\
	& & & & {\em Vert.\ Pr.} & {\em Edge Pr.} & {\em Time (s)} & {\em $\omega'$} & {\em Time (s)} & {\em $\omega'$} & {\em Time (s)} & {\em $\omega'$} & {\em Time (s)} & {\em Time (s)} & {\em Time (s)} \\
	\cmidrule(lr){1-4}
	\cmidrule(lr){5-8}
	\cmidrule(lr){9-10}
	\cmidrule(lr){11-12}
	\cmidrule(lr){13-15}
\rowcol bio-WormNet-v3-benchmark & 2 K & 79 K & 126 & 94.48 & 89.83 & 0.00383 & {\bf 126} & 0.0154 & 126 & {\bf 0.00067} & 126 & 0.00564 & 0.496 & 0.239 \\
bn-macaque-rhesus\_inter-cort-netw\_2 & 93 & 2 K & 30 & 68.82 & 82.05 & {\bf 0.000123} & {\bf 30} & 0.00024 & 30 & 0.001 & 30 & 0.00024 & 0.004 & 0.0008 \\
bn-mouse\_retina\_1 & 1 K & 91 K & 51 & 77.74 & 80.30 & {\bf 0.0125} & {\em $39^\#$} & 0.0438 & 35 & 0.184 & {\bf 51} & - & 0.26 & - \\
\rowcol cari & 1 K & 77 K & 200 & 68.28 & 11.19 & 0.783 & {\bf 200} & 0.883 & 200 & {\bf 0.176} & 200 & - & 0.656 & 4.933 \\
cavity26 & 5 K & 71 K & 19 & 98.66 & 99.00 & {\bf 0.00219} & {\bf 19} & 0.00913 & 19 & 0.0381 & 19 & 0.017 & 1.148 & 0.132 \\
econ-psmigr1 & 3 K & 411 K & 144 & 90.16 & 92.20 & {\bf 0.138} & {\em 116} & 0.728 & 114 & - & - & - & 1.38 & - \\
\rowcol frb30-15-1 & 451 & 83 K & 30 & 17.29 & 32.21 & {\bf 0.0734} & {\bf 25} & 0.133 & 25 & 1.0085 & 25 & - & 0.324 & - \\
hamming10-2 & 1 K & 519 K & 512 & 1.171 & 2.14 & 39.70 & {\bf 512} & 42.00 & 512 & {\bf 11.158} & 512 & - & 29.188 & - \\ 
light\_in\_tissue & 29 K & 188 K & 6 & 99.94 & 99.97 & {\bf 0.00773} & {\bf 6} & 0.0122 & 6 & 10.894 & 6 & 0.0306 & - & 0.63 \\
\rowcol nasa2910 & 3 K & 86 K & 36 & 96.6 & 97.48 & {\bf 0.00292} & {\bf 36} & 0.036 & 36 & 0.0272 & 36 & 0.479 & 0.64 & 0.31 \\
robot24c1\_mat5\_J & 405 & 14 K & 24 & 88.89 & 95.49 & {\bf 0.0006} & {\em 21} & 0.00283 & 20 & 0.00313 & {\bf 24} & - & 0.016 & 0.53 \\
scc\_enron-only & 152 & 10 K & 120 & 7.89 & 4.92 & 0.04 & {\bf 120} & 0.0371 & 120 & {\bf 0.002} & 120 & - & 0.336 & 0.0224 \\
\rowcol scc\_infect-dublin & 11 K & 176 K & 84 & 99.12 & 97.57 & {\bf 0.0119} & {\bf 84} & 0.0321 & 84 & 0.021 & 84 & 2.715 & - & 0.864 \\
scc\_twitter-copen & 9 K & 474 K & 581 & 87.74 & 16.80 & 26.223 & {\bf 581} & 30.967 & 581 & {\bf 10.267} & 581 & - & 38.004 & - \\
Trec12 & 3 K & 151 K & 11 & 91.71 & 98.46 & {\bf 0.00376} & {\bf 8} & 0.0381 & 8 & - & - & 21.334 & 0.756 & 8.926 \\
\rowcol polblogs & 2 K & 17 K & 20 & 95.44 & 93.74 & {\bf 0.00107} & {\em $19^\#$} & 0.00354 & 16 & 0.005 & 20 & 0.186 & 0.204 & 0.503 \\
moreno-blogs & 1 K & 1109 K & 1490 & 0.134 & 0.134 & 0.405 & {\bf 1490} & 0.369 & 1490 & 16.75 & 1490 & 1.645 & - & {\bf 0.24} \\
\bottomrule
\end{tabular}
\end{table*}
\end{landscape}

\begin{landscape}
\begin{table}[t]
\centering
\small
\caption{Performance comparison of {\em ALTHEA} on difficult real-world data.
Here, (i) results marked with $*$ denote the maximum clique size found by the heuristic before the cut-off time of $5$ min.; (ii) {\em RMC} reported segmentation fault for 
the {\em econ-psmigr1} network, and is marked with $-$; (iii) average run-times over $5$ runs are shown in sec.; and (iv) vertex and edge pruning (Pr.) are given in \% of $|V|$ and 
$|E|$, respectively.}
\label{tab:real_large}
	\begin{tabular}{lllllllllllll}
	\toprule
	\multicolumn{3}{c}{\bf Dataset Characteristics} & \multicolumn{4}{c}{\bf ALTHEA + FMC(H)} & \multicolumn{2}{c}{\bf ALTHEA + MoMC} & \multicolumn{2}{c}{\bf FMC(H)} & \multicolumn{2}{c}{\bf RMC} \\
	{\em Instance} & {\em $|V|$} & {\em $|E|$} & {\em Vert.\ Pr.} & {\em Edge Pr.} & {\em Time (s)} & {\em $\omega'$} & {\em Time (s)} & {\em $\omega'$} & {\em Time (s)} & {\em $\omega'$} & {\em Time (s)} & {\em $\omega'$} \\
	\cmidrule(lr){1-3}
	\cmidrule(lr){4-7}
	\cmidrule(lr){8-9}
	\cmidrule(lr){10-11}
	\cmidrule(lr){12-13}
\rowcol bio-WormNet-v3 & 16K & 763K & 96.23 & 91.85 & {\bf 0.0914} & {\em 94} & {\em 0.166} & {\bf 121} & 0.465 & 90 & 0.156 & {\bf 121} \\
brock800-1 & 801 & 208K & 35.21 & 58.08 & {\bf 0.088} & 17 & - & {\bf 19}* & 0.289 & 17 & - & 17* \\
C1000-9 & 1001 & 450K & 7.69 & 14.64 & 2.41 & 51 & - & {\bf 59}* & 3.065 & 51 & - & 53* \\
\rowcol econ-psmigr1 & 3141 & 411K & 90.16 & 92.20 & {\bf 0.138} & {\em 116} & {\em 0.615} & {\bf 122} & 0.728 & 114 & - & - \\
frb30-15-1 & 451 & 83K & 17.29 & 32.21 & {\bf 0.0734} & 25 & 0.866 & {\bf 29} & 0.133 & 25 & - & 25* \\
frb50-23-5 & 1151 & 581K & 16.33 & 29.77 & {\bf 2.1} & {\em 42} & - & {\bf 48}* & 3.235 & 41 & - & 40* \\
\rowcol frb53-24-5 & 1273 & 714K & 5.81 & 11.35 & {\bf 3.863} & 42 & - & {\bf 49}* & 4.65 & 42 & - & 42* \\
p-hat1500-3 & 1501 & 847K & 17.19 & 30.40 & {\bf 2.725} & 60 & - & {\bf 91}* & 4.589 & 60 & - & 60* \\
bn-mouse\_retina\_1 & 1123 & 91K & 77.74 & 80.30 & {\bf 0.0125} & {\em 39} & {\em 0.026} & {\bf 51} & 0.0438 & 35 & 0.184 & {\bf 51} \\
\rowcol polblogs & 1491 & 17K & 95.44 & 93.74 & {\bf 0.0011} & {\em 19} & 0.197 & {\bf 20} & 0.0035 & 16 & 0.005 & {\bf 20} \\
\bottomrule
\end{tabular}
\end{table}
\end{landscape}

\subsection{Experimental Evaluation}
\label{subsec:expt}

\paragraph{Baselines} We benchmark the performance of {\em ALTHEA} against the following existing state-of-the-art approaches: (i) \texttt{igraph}~\cite{igraph} C library's implementation of the exact modified Bron-Kerbosch algorithm~\cite{bron}, (ii) {\em MoMC}~\cite{momc} employing a branch-and-bound pruning strategy, (iii) {\em FMC(E)}~\cite{fmc} using exact hierarchical pruning strategy and (iv) {\em FMC(H)} -- the fast 
heuristic variant of FMC(E), and (v) {\em RMC}~\cite{rmc} -- randomized heuristic based on ``binary search'' with optimum-bounding and is obtained from the authors.
By definition, the exact algorithms give the maximum clique size~$\omega$.

Note that the final pruned subgraph obtained by {\em ALTHEA} is presented to a maximum clique solver. We couple {\em ALTHEA} with either the exact {\em MoMC} solver, or the fast {\em FMC(H)} heuristic (denoted as {\em ALTHEA+MoMC} and 
{\em ALTHEA+FMC(H)}, respectively). The approaches are evaluated on run-time efficiency and accuracy of extracting a maximum clique.
Our implementation of {\em ALTHEA} is in C, and all experiments are run on an Intel Xeon E5-2680 CPU (2.80~GHz) with 8 cores and 32 GB of RAM.

\subsection{Real Datasets}
\label{ssec:real}

We experiment on structured datasets from diverse domains such as biological networks, financial graphs, social interaction and blog conversations.
Again, our instances are obtained from Network Repository~\cite{Rossi2015}.

\paragraph{Easy instances} We selected 17 dense graphs (see Table~\ref{tab:real_small}) with varying sizes of upto 30~K vertices and 1~M edges.
Table~\ref{tab:real_small} reports the vertex and edge pruning achieved by {\em ALTHEA+FMC(H)} in addition to run-time and the maximum clique size  
extracted. 
We see that {\em ALTHEA} is highly accurate 
in identifying regions that contain a maximum clique. In fact, it is successful in extracting an optimal maximum clique in 13 of the instances, while in the remaining 4 instances, 
it extracts larger cliques than standalone {\em FMC(H)}.

We observe that {\em ALTHEA} aggressively prunes the search space (with high accuracy), achieving vertex and edge prunings as high as 99~\% 
--- with more than 80~\% vertex/edge pruning on 11 instances. This enables 
our framework to be very efficient in practice, showcasing consistent speedups of around  $3\times$ compared to the best 
performing heuristic and upto $10\times$ with respect to the exact algorithms. 
On the other hand, {\em RMC} is able to extract the maximum clique size in nearly all the instances, but suffers from large run-time in general (compared to other heuristics), owing to its dependency on  vertex coloring and independent set computation.

\paragraph{Hard instances} We select 8 additional hard instances, on which exact algorithms were unable to run to completion with a timeout of 5~minutes. 
Table~\ref{tab:real_large} tabulates these instances and the performance of the competing approaches. 
Here, we also evaluate the performance of {\em ALTHEA} when coupled with the exact {\em MoMC} solver. 

Similar to our previous observations from Table~\ref{tab:real_small}, 
we find that {\em ALTHEA+FMC(H)} performs better that the standalone {\em FMC(H)} heuristic, and extracts {\bf better solutions}. Further, 
vertex and edge pruning (of around 40~\% on average) gives {\em ALTHEA} faster run-times with upto $5\times$ speedups over {\em FMC(H)}. 
Again, {\em RMC} requires high computation time but extracts larger cliques.

From Table~\ref{tab:real_large}, we see that the pruning strategy of {\em ALTHEA} with {\em MoMC} provides an interesting trade-off between solution quality and run-time. 
This approach is able to identify significantly better solutions compared to others, in all instances. In fact, for the last two instances in Table~\ref{tab:real_large} (also 
in Table~\ref{tab:real_small}), we are now able to extract the optimal solution. Although {\em ALTHEA+MoMC} consumes slightly more run-time 
(than {\em FMC(H)}), it is still faster than {\em RMC}.

To summarize, we see that {\em ALTHEA} provides an efficient and robust pruning strategy for finding an approximate maximum clique with high accuracy in dense real-life graphs from diverse domains.
Further, as is well-known, such dense instances constitute a major challenge for state-of-the-art solvers.

\subsection{Synthetic datasets}
\label{ssec:synt}

We turn to study the robustness of {\em ALTHEA} on Erd\H{o}s-R\'{e}nyi (ER) random graphs, denoted as $G(n,p)$, which is an $n$-vertex 
graph where every edge is present with independent probability~$p$.
We observe the pruning ratio, run-time and accuracy of the approaches, by varying the two parameters $n$ and $p$.
Particularly for $p \geq 0.5$, random graphs present a challenging benchmark for pruning.
Hence, we relax the accuracy measure by considering a heuristic {\em accurate} if the size of the clique returned is at most~1 less than the optimum.

\begin{figure*}
	\centering
	\begin{subfigure}[b]{0.475\textwidth}
		\centering
		\includegraphics[width=\textwidth]{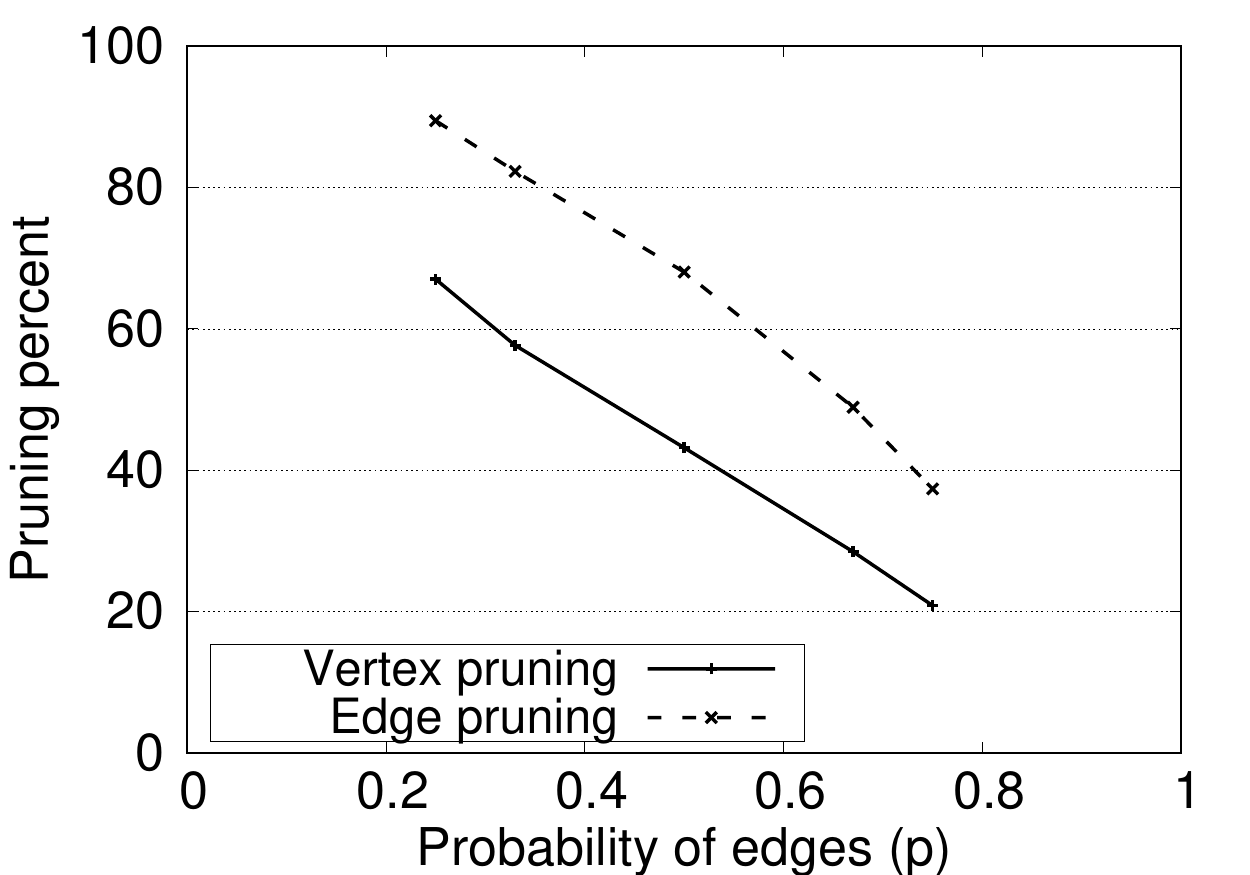}
		\caption[]%
		{}    
	\end{subfigure}
	\hfill
	\begin{subfigure}[b]{0.475\textwidth}  
		\centering 
		\includegraphics[width=\textwidth]{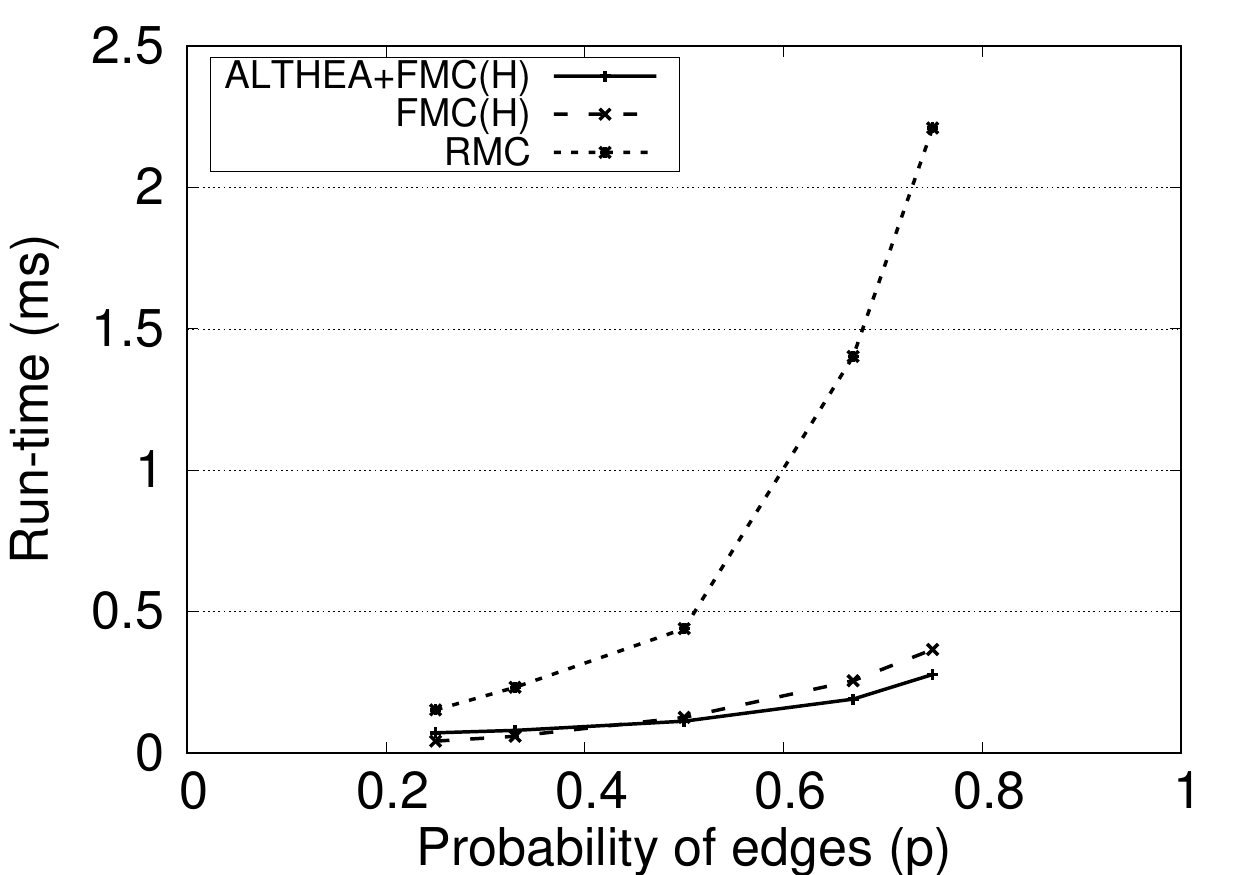}
		\caption[]%
		{}    
	\end{subfigure}
	\vskip\baselineskip
	\begin{subfigure}[b]{0.475\textwidth}   
		\centering 
		\includegraphics[width=\textwidth]{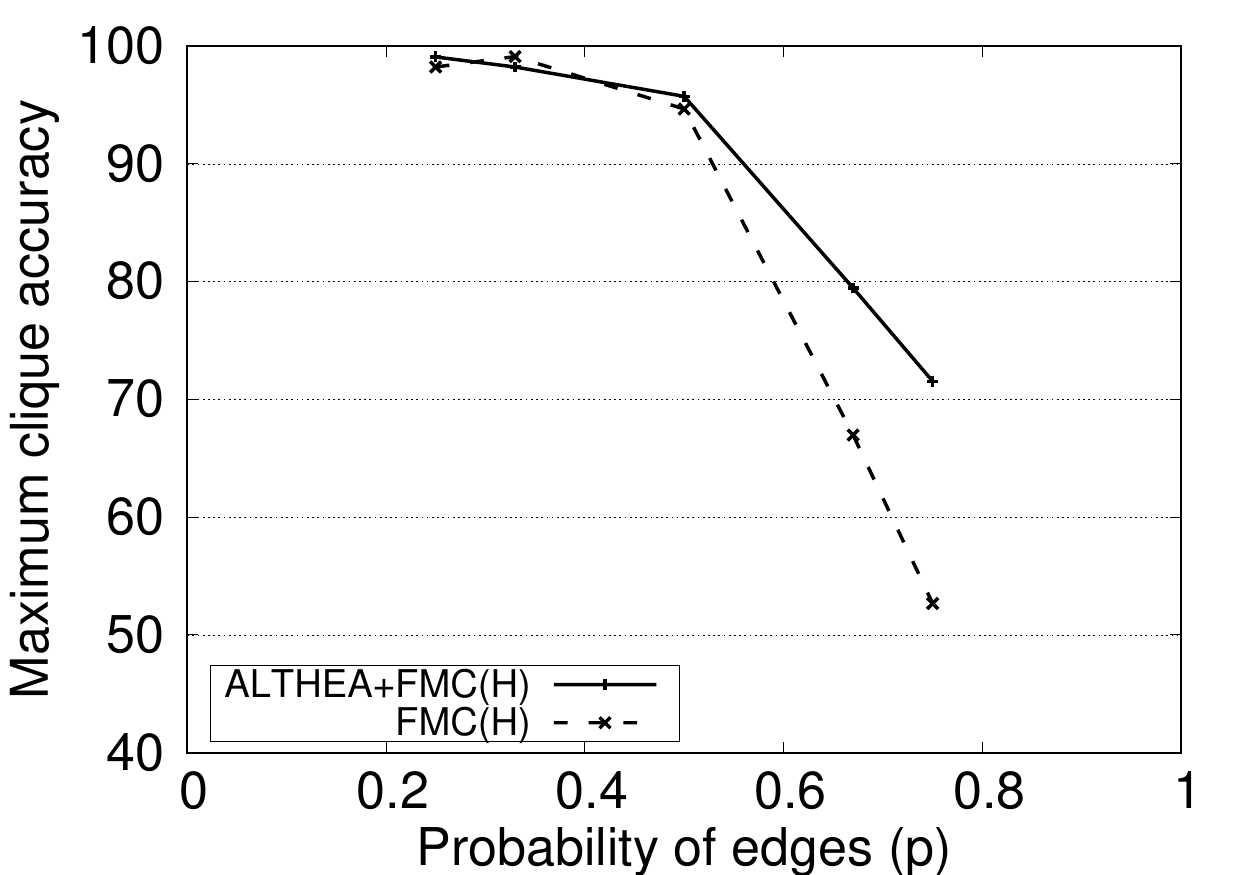}
		\caption[]%
		{}    
	\end{subfigure}
	\quad
	\begin{subfigure}[b]{0.475\textwidth}   
		\centering 
		\includegraphics[width=\textwidth]{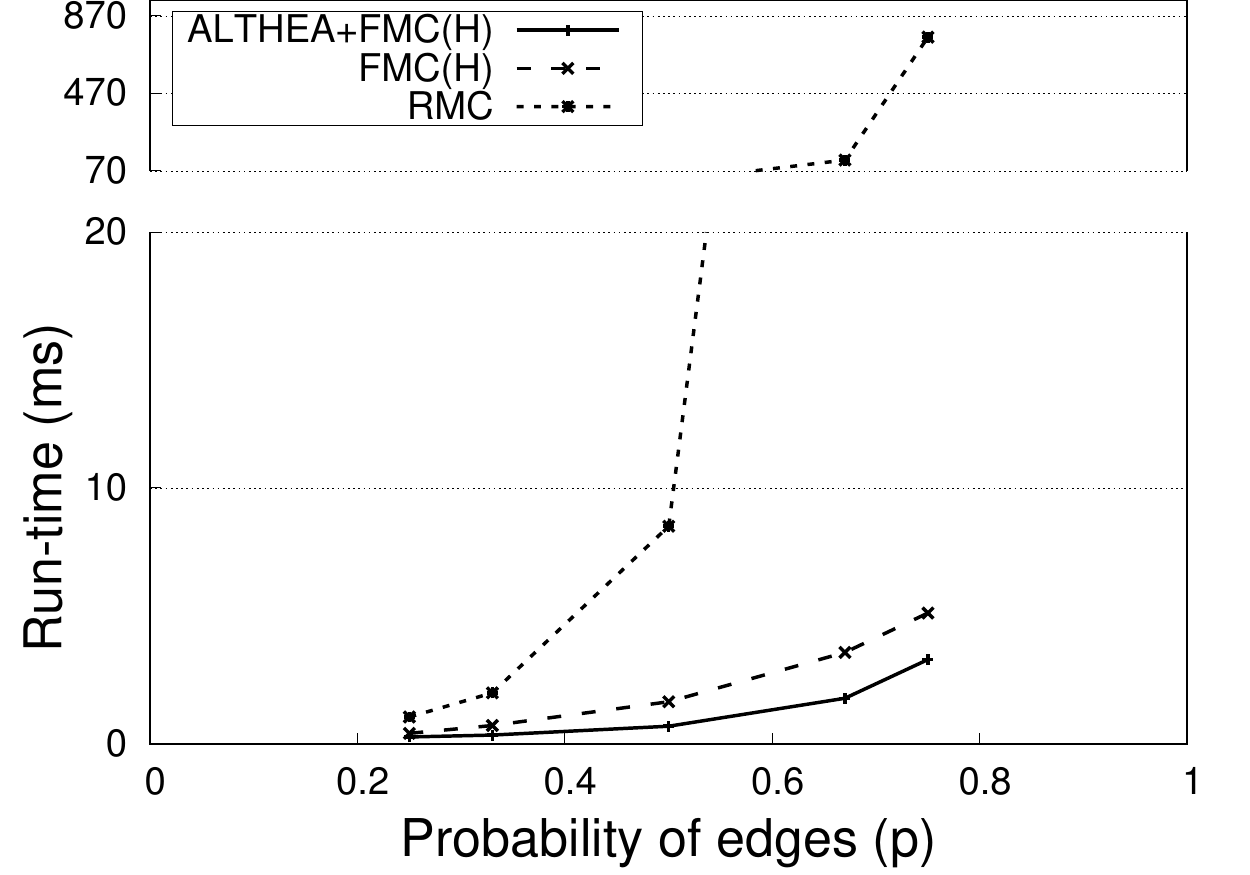}
		\caption[]%
		{}    
	\end{subfigure}
	\vskip\baselineskip
	\begin{subfigure}[b]{0.475\textwidth}   
		\centering 
		\includegraphics[width=\textwidth]{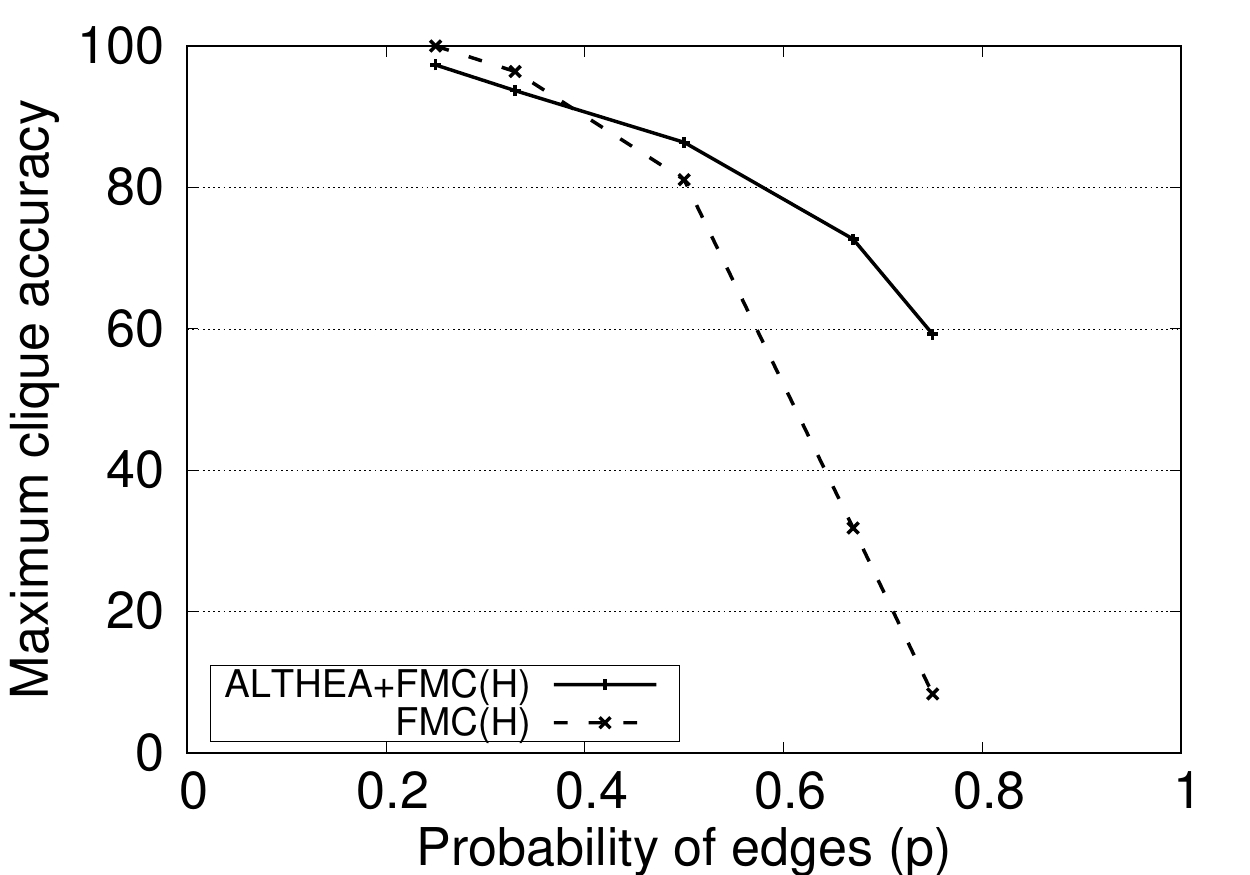}
		\caption[]%
		{}    
	\end{subfigure}
	\caption[ The average and standard deviation of critical parameters ]
	{Performance comparison of competing approaches on ER random graphs $G(64,p)$ with varying density $p$ based on \textbf{(a)} vertex and edge pruning rates, \textbf{(b)} run-time, and 
	\textbf{(c)} maximum clique accuracy; and on $G(256,p)$ with varying density $p$, based on \textbf{(d)} run-time and \textbf{(e)} maximum clique accuracy.} 
	\label{fig:er_256}
\end{figure*}

\paragraph{Graph density} The effect of density on the performance of the approaches is shown in Figures~\ref{fig:er_256}(a)-(c) 
obtained on ER-graphs with~64 vertices with varying density of $p \in \{0.25,0.33,0.5,0.67,0.75\}$. 
In terms of pruning rate, we observe in Figure~\ref{fig:er_256}(a) that {\em ALTHEA} effectively prunes nearly 50~\% of the edges (and vertices) even in dense random graphs (i.e., $p=0.5$). However, the pruning rate decreases linearly with increase in density (to around 20~\% for $p=0.75$). The high pruning rate enables {\em ALTHEA} (coupled with FMC(H) heuristic) to be superior than the other approaches in terms 
of run-time demonstrating upto $1.5\times$ speedups compared to the standalone {\em FMC(H)}. 
Similar to the real datasets, 
{\em RMC} suffers from high run-time (upto $10\times$ slower). Interestingly, we observe that {\em ALTHEA} 
exhibits higher accuracy compared to {\em FMC(H)} (Figure~\ref{fig:er_256}(c)). For $G(64,0.75)$, we report an accuracy
of more than 70~\% compared to around 50~\% for {\em FMC(H)}.
The accuracy of {\em FMC(H)} is seen to degrade significantly as density increases. 
For low density graphs (i.e., $p < 0.5$), both heuristics perform similarly.
{\em RMC} has perfect accuracy, but infeasible running times for larger and denser graphs.

\paragraph{Graph size} We assess the effect of varying $n$ on the performance of {\em ALTHEA}. 
Figures~\ref{fig:er_256}(d)-(e) present the results for $n=256$.

The approaches are seen to exhibit similar behaviour as above, with high pruning rates for {\em ALTHEA}, along with a large speedup in extracting large cliques 
compared to {\em FMC(H)}. 
From Figure~\ref{fig:er_256}(e), we observe that our approach 
depicts significantly superior accuracy (compared to {\em FMC(H)}) -- being nearly $6\times$ more accurate in identifying a maximum clique in dense input graphs.
Finally, we remark that similar results were observed on ER-graphs for other parameter values of $n$ and $p$, but omit further details.

\section{Conclusions}
We have proposed a novel framework for learning to scale-up combinatorial optimization algorithms. In contrast to the existing learning frameworks that use difficult-to-interpret learning models to learn the exact decision boundary, our proposed framework relies on interpretable learning models with local features to prune the elements that are not in any optimal solution(s). The deeper insights learned by our multi-stage pruning framework result in the identification of feature combinations relevant to the optimization problem and the instance class. This can result in better heuristics for the problem, as evidenced by maximum clique enumeration.

Our framework has been designed primarily for combinatorial optimization problems that involve finding an optimal subset of elements. A crucial direction for future research is to explore if this framework can be extended to deal with combinatorial optimizations involving ordering and assignment problems. Other avenues for future research include the design of approaches to improve the accuracy of the pruning by incorporating problem constraints in the learning process. 
\bibliographystyle{abbrv}
\bibliography{bibliography}
\end{document}